\documentclass[11pt]{article}
\usepackage{booktabs}

\def\math#1{$#1$}

\def\mand#1{$$#1$$}
\def\mandc#1{\mand{\abovedisplayskip=3pt plus 1pt minus 1pt%
\abovedisplayshortskip=0pt plus 1pt minus 1pt%
\belowdisplayskip=3pt plus 1pt minus 1pt%
\belowdisplayshortskip=0pt plus 1pt minus 1pt%
#1}}

\def\frac#1#2{{#1\over #2}}

\def\mld#1{\begin{equation}
#1
\end{equation}}
\def\mldc#1{\mld{\abovedisplayskip=3pt plus 1pt minus 1pt%
\abovedisplayshortskip=0pt plus 1pt minus 1pt%
\belowdisplayskip=3pt plus 1pt minus 1pt%
\belowdisplayshortskip=0pt plus 1pt minus 1pt%
#1}}

\def\eqan#1{\begin{eqnarray*}
#1
\end{eqnarray*}}

\DeclareSymbolFont{AMSb}{U}{msb}{m}{n}
\DeclareMathSymbol{\N}{\mathbin}{AMSb}{"4E}
\DeclareMathSymbol{\Z}{\mathbin}{AMSb}{"5A}
\DeclareMathSymbol{\R}{\mathbin}{AMSb}{"52}
\DeclareMathSymbol{\Q}{\mathbin}{AMSb}{"51}
\DeclareMathSymbol{\I}{\mathbin}{AMSb}{"49}
\DeclareMathSymbol{\C}{\mathbin}{AMSb}{"43}

\def\choose#1#2{\left({{#1}\atop{#2}}\right)}

\def\Z{{\mathbf Z}}

\def\floor#1{{\left\lfloor\,#1\,\right\rfloor}}
\def\ceil#1{{\left\lceil\,#1\,\right\rceil}}
\def\r#1{{(\ref{#1})}}

\newcounter{exercisenum}

\def\dotfil{\leaders\hbox to 1.5mm{.}\hfill}
\newcounter{rmnum}
\def\RN#1{\setcounter{rmnum}{#1}\uppercase\expandafter{\romannumeral\value{rmnum}}}
\def\rn#1{\setcounter{rmnum}{#1}\expandafter{\romannumeral\value{rmnum}}}

\def\Prob{\mathbb{P}}
\def\Exp{\mathbb{E}}

\def\mumin{q}
\newcommand{\remove}[1]{}
\usepackage{tikz,bm}
\usepackage{url}
\usepackage{enumitem}
\usepackage{times}
\usepackage{natbib}
\usepackage{graphicx}
\usepackage{amsmath,amssymb,amsthm}
\usepackage{rotating}
\usepackage{multicol}
\usepackage{latexsym}
\usepackage{verbatim}
\usepackage{color}
\definecolor{lightgrey}{rgb}{0.85,0.85,0.85}
\usepackage{lineno}
\usepackage[ruled]{algorithm2e}
\usepackage{float}
\newcounter{newct}

\newcommand{\bv}{\begin{array}}

\newcommand{\cost}[1]{\text{Cost}(#1)}
\newcommand{\ben}[1]{\text{Ben}(#1)}
\newcommand{\sw}[1]{\text{SW}(#1)}

\newcommand{\correct}{R}
\newcommand{\correctM}{Q}

\newtheorem{thm}{Theorem}
\newtheorem{dfn}{Definition}

\newtheorem{lem}{Lemma}
\newtheorem{ex}{Example}

\newtheorem{coro}{Corrollary}

\newcommand{\ra}{\rightarrow}

\newcommand{\Omit}[1]{}

\renewcommand{\bm}[1]{{\bf\rm #1}}

\usepackage[top=1in, bottom=1in, left=1in, right=1in]{geometry}

\begin{document}
\title{%
%A Mathematical Model For Optimal Decisions In A Representative Democracy 
A Mathematical Model For Optimal Decisions In A Representative Democracy 
%: Let Many But  Wiser Representatives Vote for You
%  More Representatives Elected from Smaller Groups Wins
%  Let The Better-Informed Vote for You, But We Need More Representatives
%:\\ Optimal Group Size for Better Representative Democracy
}
\author{ Malik Magdon-Ismail\\ RPI\\ magdon@cs.rpi.edu
 \and Lirong Xia \\ RPI\\ xial@cs.rpi.edu}
\footnotetext{Rensselaer Polytechnic Institute, \{magdon,xial\}@cs.rpi.edu}
\date{}
\maketitle

\begin{abstract}
  Direct democracy is a special
  case of an ensemble of classifiers,
  where every person (classifier) votes on every issue. This fails when the
  average voter competence (classifier accuracy) falls below 50\%, which 
  can happen in
  noisy settings where voters have only limited information, or when
  there are multiple topics and the average voter competence may not
  be high enough for some topics.
  Representative democracy, where voters choose representatives to vote,
  can be an
  elixir in both these situations.
  Representative democracy is a specific way to improve the ensemble of
  classifiers.
  We introduce a mathematical model for
  studying representative democracy, in particular understanding the parameters
  of a representative democracy that gives maximum decision
  making capability. Our main result states that under general and natural
  conditions,
  \begin{enumerate}
    \vspace*{-2pt}
  \item Representative democracy can make the correct decisions simultaneously for
    multiple noisy issues.
  \item When the cost of voting is fixed, the optimal representative democracy requires that
    representatives are elected from  constant sized groups:
    the number of representatives should be linear in the number of voters.
    \item When the cost and benefit of voting are both polynomial, the optimal group size is close to linear in the number of voters. 
   \vspace*{-2pt}
  \end{enumerate}
  This work sets the mathematical
  foundation for studying the quality-quantity tradeoff
  in a representative democracy-type ensemble
  (fewer highly qualified representatives
  versus more less qualified representatives). 
\end{abstract}

% !TEX root = wiser.tex
\section{Introduction}
Suppose a voter-population of size \math{n}
must vote in a referendum to make an important binary decision to optimize some objective, e.g.~social welfare growth over 10 years.
A typical solution is \emph{direct democracy} which decides based
on a majority
vote, the so called ``wisdom of the crowd.'' Direct democracy works when,
the crowd does indeed have wisdom.
In reality, the voters cannot directly observe which decision is correct. Instead, they form beliefs using perceived information, which can be inaccurate,
misinterpretable or even manipulated. For example, suppose that each voter's chance to vote for the correct decision, her {\em competence},
is i.i.d.~generated uniformly over $[0,0.99]$. Now, the majority
among many
voters makes the wrong decision with near
certainty~\citep{Fey03:Note}.

This example highlights a critical flaw of direct democracy, where voters
participate in decision-making irrespective of their competence. Direct democracy is particularly problematic in high noise situations, where there is often a close tie in people's beliefs between two choices, as in the previous example.
%, or even in situations where the average competence is slightly larger than $0.5$, because the number of voters may not be large enough to guarantee the correctness of majority voting.
Unfortunately, such close
ties are common in real-life high-stakes scenarios. For example, in the 2016 United Kingdom European Union membership referendum, 51.89\% voted for leave and 48.11\% voted for remain. In the 2016 US Presidential Election, 46.1\% voted for Trump and 48.2\% voted for Clinton.\footnote{These examples are only used to  show real situations where close ties exist. We do not know if direct democracy would succeed or fail in these cases since we do not know what the
``correct'' outcome is.} How can democracy cope with high-stakes noisy
issues where the average voter competence may drop below \math{0.5}, especially
if there is misinformation?

One promising rescue is {\em representative democracy}, where voters form
groups (or districts).
Each group chooses a representative, and representatives decide via a
majority vote. The tradeoff is that there are
fewer representatives than base-voters,
but, in return, each representative is  (hopefully)
better informed, being the ``wisest'' 
from its group, or at least having a higher competence than 
the average in its group members.
%Therefore, the probability for representative democracy to make the correct decision can be higher than that of direct democracy.
Continuing the example above, let us now use representative democracy, where people are divided into households (e.g.~$5$ people per group), and let each group choose the member with highest competence as its representative. Then, with high probability the representative's competence is strictly larger than $0.5$.
Now, with enough representatives, a majority vote will now make the correct decision with near certainty~\citep{Fey03:Note}.

%There are two major types of democracy~\citep{bernardbook}: {\em direct} democracy, where all people directly participate in decision-making processes via means such as referendum voting; and {\em representative} democracy, also known as {\em indirect} democracy, where the  population is divided into groups, each group chooses a representative, and the representatives vote on multiple issues on behalves of their groups~\citep{Manin1997:The-Principles}. 

Many countries and organizations adopt a mixture of direct democracy and representative democracy. In the US, voters in each state vote on multiple referenda (direct democracy). In addition, voters elect
members of the congress (representative democracy) to make
decisions for everyone by voting on bills.
%Nearly all modern western-style democracies are types of representative democracies~\citep{wiki}. 
Representative democracy can have a {\em fixed number of representatives} regardless of the population (e.g.~the US senate has 100 members, two from each state),
or a {\em fixed group size}, where the number of representatives is proportional to the population (between 1789 and 1913, the US House increases from 65 to 435 members based on the nation's population growth\footnote{After 1913 the number is fixed to 435.}).
While it is widely accepted that representative democracy is efficient due to lower operational cost and has better turnout than direct democracy, there is still a large debate on the following fundamental question: {\em which type
of democracy makes better decisions?} 

We are not aware of quantitative answer to this question, nor a
mathematical framework for analyzing representative democracy w.r.t.~its capability to make correct decisions. This is in sharp contrast to direct democracy, which has been mathematically analyzed in depth
to provide a justification of the ``{wisdom of the crowd},''
which dates back to the {\em Condorcet Jury Theorem}
in the 18th century~\citep{Condorcet1785:Essai}.
Roughly speaking,
the Jury Theorem states that a large group of voters are likely to make a correct decision by majority voting, which ``{\em lays, among other things, the foundations of the ideology of the democratic regime}''~\citep{Paroush98:Stay}.
Direct democracy is just a representative democracy where each
group has one voter. Thus, a mathematical characterization of
optimal representative democracies would also
highlight the subcases
where direct democracy is best.
The goal of this paper is to establish rigorous mathematical foundations of representative democracy, and provide quantitative answers to the following key questions.

{\bf Key questions.}
\begin{enumerate}
\item
What is the optimal number of representatives for representative democracy?
\item How should one optimally
divide the voters into groups, with each group electing one
representative? Specifically, what is the optimal size distribution for
the groups?
\end{enumerate}
We will answer these questions in a general setting, where
the groups satisfy a weak form of homogeneity. At a high level, the
representative election process can be described by a function
\math{\mu} that takes a group as input
and outputs its representative.
By homogeneous, we mean that this representative election function
is the same in every group.
A concrete example of this paradigm is when
each group runs the same type of election process on its members who are
independent and drawn from some underlying voter-distribution.

In our analysis, 
we consider two cases. The first is when there is a fixed cost for the
voting. In this case, the goal of the representative democracy is to maximize
the chances of making the correct decision. This is case is relevant to 
making extremely important decisions where the operational cost of voting
is not considered a valid tradeoff for correctness.
The second case is when the cost of voting increases with the number
of representatives, in which case one
must balance the cost with the benefit that accrues to all \math{n} voters.

\subsection{Our Contributions}
We provide a novel mathematical model of representative democracy w.r.t.~its ability to make correct decisions. With our model, we obtain characterizations on optimal number of representatives in representative democracy, and our main messages are the following.

1. When the cost of voting is fixed, fixed group size is optimal.

2. When the cost and benefit of voting are both polynomial, $O(\log n)$ representatives are optimal, where $n$ is the number of voters.

In our basic model, there is a single binary issue to be decided.  $n$ voters are divided into $L$ groups, each group chooses a representative by a {\em representative selection function}, and the representatives will use majority voting to make a binary decision on the issue. We assume that each voter's type is characterized by her {\em competence}, which is her probability to vote for the correct decision on the issue, and is generated i.i.d.~from a distribution $F$. Let $\ben{n}\in \mathbb R$ denote the benefit of making the correct decision compared to making the wrong decision for $n$ voters, and let $\cost{L}\in \mathbb R$ denote the operational cost for $L$ representatives to vote.

We reduce the representative selection function to a {\em group competence function} $\mu:{\mathbb N}\rightarrow [0,1]$, which maps each  group size to the expected competence of the representative. %Using this simplification,
We then extend the Condorcet Jury Theorem to representative democracy by characterizing group competence functions that would choose representatives to make the correct decision with probability $1$ as $n\rightarrow\infty$. Our main theorems are surprising characterizations of the optimal  group size when the cost of voting is fixed.%, which states that when the representative selection function $\mu$ approaches to $1$ at a polynomial rate, i.e.~not faster than than $O(\frac1{K^\alpha})$, then the optimal group size is $O(\alpha \ln n+1)$.

\vspace{2mm}
{\noindent\bf Theorems~\ref{thm:constant}, \ref{thm:oneissuegeneral} and~\ref{thm:groups-homogeneous} (optimal representative democracy for single issue with fixed cost of voting, informally put).} Under natural and mild conditions on the group competence function \math{\mu} and when $\cost{L}$ is a constant:
\begin{enumerate}\vspace*{-4pt}
\item (Homogeneous groups) The optimal group size \math{K^*(n)}
is at most a constant independent of \math{n}.
\item (Inhomogeneous groups) The optimal number of representatives
\math{L^*} is linear in \math{n}.
\item The optimal group size distribution is nearly homogeneous.
\end{enumerate}
These results hold, independently of the specific details of the
representative selection process.
Let us highlight why the result is unexpected in a concrete context.
Suppose voter competence is drawn from some continuous
density on \math{[0,1]} and a group elects the most competent of its voters as
representative (very optimistic since
it cannot get better than that). Then by choosing larger and
larger groups, the competence of the representative approaches 1.
The price paid
is that there are fewer of these ultra-smart representatives, but there will
still be many of them as \math{n\rightarrow\infty}. Indeed, one might posit
that some optimal tradeoff exists whereby the group size tends to infinity
but at a slower rate than \math{n}
so that the representative gets smarter and smarter \emph{and} there are
more and more representatives. Our theorem establishes the contrary.
The optimal group size will never exceed some constant (the constant's
value may depend on the specific parameters of the selection process).

%The theorems immediately imply that representative democracies with fixed number of representatives \emph{cannot be optimal}
%(US senate and US House after 1913).
%Moreover, the group sizes should be nearly the same.
%Recall that direct democracy fails when the
%average voter competence is less than 0.5. However, representative democracy
%with a constant ``household'' size can succeed! Numerical calculations
%suggest that the optimal number of representatives needs to be significantly
%higher for the US House of representatives, which is currently set at
%$435$ for a voting population of about \math{n=\text{235 million}}.
%For a voter competence that is uniformly distributed in \math{[0.45,0.52]}
%(the average voter
%is just slightly less competent than \math{0.5}), here are the numbers:
%\begin{center}
%  {\tabcolsep5.5pt
%    \begin{tabular}{c|c|c|c|c|c|c}
%    & \math{K=1} &\math{K_*=3}&\math{K^*=9}
%  &\math{5\times\text{congress}}
%  &\math{2\times\text{congress}}
%  &\math{\text{congress}}\\\hline
%  Success&0\%&
%  100\%&100\%&97\%&88\%&80\%
%\end{tabular}
%  }
%  \end{center}
%Here, \math{K_*} is the minimum group size needed to achieve consistency (as number of votes goes to infinity, the majority of $n/K_*$ representatives is $1$ goes to $1$); \math{K^*} is the optimal group size.

To prove our results, we use
novel combinations of combinatorial bounds. In addition, some of our
results may be of independent interest, 
for example, Lemma~\ref{lem:poisson-binomial}
answers an open question on the probability
for majority voting to be correct when the average competence of voters is exactly $0.5$ where there are no general results for the
non-asymptotic behavior~\citep{Fey03:Note} and the asymptotic behavior is only
conjectured~\cite[Lemma 5]{Owen89:Proving}.

We then consider the case of polynomial cost and polynomial benefit, that is, there exist constants $q_1>0$ and $q_2>0$ such that $\cost{L}=L^{q_1}$ and $\cost{n}=n^{q_2}$. A special case is linear cost and linear gain, where $q_1=q_2=1$. It turns out that the cost of voting has a significant impact on the optimal group size.

\vspace{2mm}
{\noindent\bf Theorems~\ref{thm:constantlinearcost} (optimal representative democracy for single issue with polynomial cost and polynomial benefit, informally put).} Under natural and mild conditions on the group competence function \math{\mu} and suppose there exist constant $q_1>0$ and $q_2>0$ such that $\frac{\cost{L}}{\ben{n}}=\Theta(\frac{L^{q_1}}{n^{q_2}})$:
\begin{enumerate}\vspace*{-4pt}
\item The optimal homogeneous group size 
is $\Omega(n/\log n)$.
\item When $\mu(K)$ converges to $1$ at a polynomial rate, the  optimal homogeneous group size 
is $\Theta(n/\log n)$
\item If $\mu(K)$ is upper bounded away from $1$, then the optimal group size is $\Theta(n)$.
\end{enumerate}

Finally, we  extend our analysis to the situation where the representative
will have to vote on $d>1$ possibly correlated issues.
This is the case with the US Senate and House members who after election
may have to cast a vote on multiple occasions on different issues (about once
per week).
Now,
the group competence function  $\mu$ must be
extended to a mapping from group size $K$ to a $2^d$-dimensional multinomial distribution. We extend the Condorcet Jury Theorem to this setting, and our main theorem states that the optimal group size is increased by $\Theta(\ln d)$.

\vspace{2mm}
{\noindent\bf Theorem~\ref{thm:multiissue} (optimal representative democracy
for multiple issues with fixed cost, informally put).} Under natural and mild
conditions on the group competence function, the optimal
group size is \math{O(\ln d)} and the optimal
number of representatives is $\Omega(n/\ln d)$ for $d\ge 2$ issues.
%or $\Omega(n/(\ln d+ \ln n))$. 
%The bounds are asymptotically tight in some cases.
\vspace{2mm}

With multiple issues, there is no obvious 
way to elect representatives to achieve consistency
in \emph{every issue}.
For example, when issues are independent,
if a group picks its representative according to its favorite (or random)
issue, the results can be disastrous.
We introduce a natural representative selection process for multiple issues,
the  {\sc max sum} process, which chooses the group member with the maximum
sum of competences across all issues. This mimics a small group (e.g.~a household) choosing its most ``informed'' member. We show that under natural assumptions and for independent issues, {\sc max sum} is
consistent on \emph{all issues}.
%So, {\bf let many but wiser representatives vote for you.}

\subsection{Related Work and Discussions}
There have been large voices in the US in favor of increasing the number of House representatives, as nowadays each of them represents roughly 700K people, and such number used to be as low as 33K. The arguments are mostly from the perspective of representation rights and the intention of framers of the Constitution and the Bill of Rights~\citep{CNNexpandcongress,NYTenlarge,WashingtonPexpand}. More information and activities can be found at \url{http://www.thirty-thousand.org}. Our work provide a mathematical foundation to analyze the rationale behind this move.

While mathematical models of representative democracy are not new, e.g.~\citep{Besley1997:An-Economic,Auriol2012:On-the-optimal}, we are not aware of previous work that quantitatively characterizes optimal representative democracy w.r.t.~its ability to make correct decisions. Our work is related to the literature on extensions of The Condorcet Jury Theorem to heterogeneous agents, where voters may have different competences~\citep{Nitzan80:Investment,Grofman83:Thirteen,Nitzan84:Significance}. More previous work along this vein
is in one of the three subareas: (1)  understanding the conditions for consistency of majority voting~\citep{Owen89:Proving,Paroush98:Stay,Kanazawa98:Brief,Fey03:Note,Sapir05:Generalized}, (2) studying optimal population size to maximize the correctness of majority voting~\citep{Feld84:Accuracy,Miller86:Information,Gradstein87:Organizational,Paroush89:Robustness,Berend98:When,Ben00:Nonasymptotic,Mukhopadhaya03:Jury,Karotkin03:Optimum,Berend05:Monotonicity,Berend07:Monotonicity,Stone2013:Optimal,Ben11:Condorcet}, and (3) incentivizing voters to increase their competence~\citep{Nitzan80:Investment,Karotkin95:Incentive,Ben03:Investment}.  See~\citet{Nitzan17:Collective} for a recent survey, including a nice overview of extensions of the Condorcet Jury theorem to dependent voters and strategic voters.

%The work that is closest to ours is by~\citet{Fey03:Note}, who assumed that voters' competence is generated i.i.d., and proved that when the average competence is strictly larger (respectively, smaller) than $1/2$, the majority rule is consistent (respectively, is not consistent). It is an open question when the average competence is exactly $1/2$, whether the majority rule is consistent, which we will give a negative answer in this paper.

Our work on single issue is related to the first subarea (consistency) and the second subarea (optimal size). The key differences are, first, in our work, the competence of the representatives are computed endogenously as a result of partitioning and representative selection, while the competence of voters are given as input in the first subarea above.  Second, in our work, increasing the group size will reduce the number of groups, thus may reduce the correctness of majority voting, even though each representative has a higher competence, while in the second subarea above there is one group of voters with variable size, who directly vote on the issue. In other words, our framework allows for a quantitative analysis of the quality vs.~quantity tradeoff in representative democracy. Our work on multiple issues  is significantly different from previous work.

Our work is related but different to two recent works  on liquid democracy (a.k.a.~delegative democracy). \citet{Cohensius2017:Proxy} studied
 proxy voting, where a set of proxies are selected either randomly or voluntarily, and each voter chooses a proxy to vote for her. The authors compared the accuracy of proxy voting theoretically and experimentally, and identify conditions for proxy voting to be superior than direct democracy.  \citet{Kahng2018:Liquid}  studied whether delegation mechanisms on networks can improve direct democracy w.r.t.~the probability to reveal the ground truth. The main difference between our work and the two paper is the dynamic of decision-making process, and therefore, the main questions are different. In our work (representative democracy), voters are divided into multiple groups, and each group chooses a representative to vote for the group members.  In the work of~\citet{Cohensius2017:Proxy}, proxies are chosen first, and are then weighted by their popularity. In the work of~\citet{Kahng2018:Liquid}, each voter can delegate her vote and all votes delegated to her, to another voter with higher competence. We note that operational cost of voting was also not considered in \citep{Cohensius2017:Proxy} and~\citep{Kahng2018:Liquid}.

Our work is also related to recent research in computational social choice that uses statistics to make better decisions~\citep{Conitzer05:Common,Caragiannis2016:When,Xia11:Maximum,Young88:Condorcet,Procaccia12:Maximum,Pivato13:Voting,Elkind14:Electing,Azari14:Statistical,Xia2016:Bayesian}. These works focus
on direct democracy, not representative democracy.

Lastly, our results on multiple issues is related to multi-issue voting~\citep{Koriyama2013:Optimal,Skowron2015:What,Lang16:Voting,Conitzer2017:Fair}. The main difference is that the consideration in multi-issue voting is often fairness and representativeness, rather than making correct decisions.

\section{Modeling Representative Democracy}
In this section we propose a mathematical model for representative democracy for one issue. We will extend it to multiple issues in Section~\ref{sec:multi}. As in the Condorcet Jury Theorem, we assume that voters' ability to vote for the correct decision is drawn i.i.d.~from a distribution $F$. The voters are divided into $L\ge 1$ groups. For any group $\ell$ with $K$ voters, whose competences are $\{q_{\ell,1},\ldots,q_{\ell,K}\}$, a deterministic or randomized representative selection process $V_\ell$  chooses a member $q_\ell$. The chosen representative casts a vote, and majority voting succeeds if strictly more than half of representatives vote $1$. The process for a group to choose a representative to vote is summarized below.
\mandc{
F\quad
  \tikz[>=latex]{\path(0,0)edge[->,line width=1pt]node[above=-1pt,font=\scriptsize]{generate voters}(2,0)}\quad
  \{q_{\ell,1},\ldots,q_{\ell,K}\}\quad
  \tikz[>=latex]{\path(0,0)edge[->,line width=1pt]node[above=-1pt,font=\scriptsize]{elect representative}(2.5,0)}\quad
  q_\ell\quad
  \tikz[>=latex]{\path(0,0)edge[->,line width=1pt]node[above=0pt,font=\scriptsize]{cast vote}(1.25,0)}\quad
  x_\ell
}

Formally, our basic model is defined as follows. 

\begin{dfn} A representative democracy for one issue is composed of the following component.

$\bullet$ {\bf Issue.} Suppose there is one issue to decide, whose outcome is $1$ (correct), and $0$ incorrect.

$\bullet$ {\bf Partition function $\vec K$.} For any number of voters $n$, $\vec K$ denote a partition function that divides $n$ voters into $L(n)$ groups, that is, $\vec K(n)=(K_1,\ldots,K_{L(n)})> 0$ and $\sum_{i=1}^{L(n)}K_i=n$.

$\bullet$ {\bf Distribution of competence $F$.} We assume that each voter's {\em type} is characterized by her {\em competence}, which is the probability for her to vote for $1$. Each voter's competence is i.i.d.~generated from a distribution $F$ over $[0,1]$.

$\bullet$ {\bf Representative selection process $V_\ell$.} Suppose group $\ell$ has $K$ users whose competences are $q_{\ell,1},\ldots,q_{\ell,K}$ respectively. The group uses a (randomized) process $V_\ell(q_{\ell,1},\ldots,q_{\ell,K})$ to select a representative $q_\ell$, whose vote is represented by a random variable $x_\ell$.

$\bullet$ {\bf Voting by the representatives.} The chosen representatives will vote for $1$ with probability equivalent to her competence, and vote for $0$ otherwise. The majority rule is used to decide the outcome based on representatives' votes. We assume that a strict majority of votes is necessary to make the correct decision. That is, when there is a tie, the outcome is $0$ (incorrect).
\end{dfn}

When each group has exactly $1$ member, we obtain the
direct democracy, otherwise we have the representative democracy as in the following example.

\begin{ex}[Uniform Voters with Uniform Process]\label{example:uniform-uniform}\rm
  Suppose there are $L\ge 1$ groups, each group has $K\ge 1$ members. Each voter's competence is i.i.d.~generated from a uniform distribution $F=\text{Uniform}[a,b]$. For all groups, $V_\ell$ chooses a member uniformly at random. It is not hard to see that majority voting by the representatives is correct with probability no more than $0.5$ when $a+b<1$.
  \hfill$\blacksquare$
\end{ex}

%\begin{ex} Suppose each voter's competence is drawn i.i.d.~from the uniform distribution over $[0,1]$. The following three methods are natural ways to select the representative for each group.
%Method 1: random pick. The representative is chosen uniformly at random.
%Method 2: Perfect most-informed. The representative is the group member with highest competence. Often the competence is not directly observable, but we assume that the group members, such as members of a household have some idea who is most informed.
%Method 3: Noisy most-informed. Given $(p_1,\ldots,p_K)$ such that $p_1\ge p_2\ge\cdots\ge p_K\ge  0$ and $\sum_{i=1}^K p_i=1$. The representative is chosen to be the voter with the $i$-th highest competence with probability $p_i$. This is generalization of Method 2, where the representative is chosen with probability that are correlated with her competence.
%\end{ex}

We note that the vote of group $\ell$'s representative, i.e.~the random variable \math{x_\ell}, is characterized entirely by
its expectation, which contains all information needed in the analysis in this paper. Therefore, we will simplify the representative selection process to a single {\em group competence function} $\mu$, which specifies the expected competence of the representative as a function of the group size $K$. 
%\math{\mu(K)} which computes \math{\Exp[x_\ell]} as a function of \math{K} under $F$.
\begin{dfn}\label{def:rep-dem}
For one issue, a representative democracy is a {\em selection function}
  \math{\mu:\N\mapsto [0,1]}.
\end{dfn}
For example, the selection function for the uniform process in Example~\ref{example:uniform-uniform} can be represented by $\mu_\text{U}(K)=(a+b)/2$ for all $K$. We will see that the selection function significantly simplifies the process in the following two examples. In the next example, the group chooses its most-informed member---the one with the highest competence---as the representative. 
\begin{ex}[Max Process]\label{example:uniform-max}
\rm
As in Example~\ref{example:uniform-uniform}, suppose $F=\text{Uniform}[a,b]$. Each group now chooses the member with the highest competence as the representative. So \math{q_\ell=\max\{q_{\ell,1},\ldots,q_{\ell,K}\}}. In other words, $q_\ell$ is the $K$-th order statistic of $K$ uniformly distributed random variables, which means that $\frac{q_\ell-a}{b-a}$ is a Beta random variable with PDF $\text{Beta}(K,1)$~\cite{Gentle2009:Computational}, whose mean is $\frac{K}{K+1}$. We call this representative democracy the {\sc max} process, denoted by $\mu_{\max}$. Therefore,
$$\mu_{\max}=\Exp[x_\ell]= \frac{K}{K+1}\times (b-a)+a=\frac{1}{K+1}a+\frac{K}{K+1}b$$
\hfill$\blacksquare$
\end{ex}
%\begin{proof}
% We compute \math{\mu(K)} using
%  the law of iterated expectation,
%  \mand{\Exp[x_\ell]=\Exp_{q_\ell}[\Exp[x_\ell\mid q_\ell]]
%    =\Exp_{q_\ell}[q_\ell]=\int_{a}^bds\ s f_{q}(s)=
%    k\int_{a}^bds\ s f_p(s)F_p(s)^{k-1},}
%  where \math{f_q} is the probability density of \math{q_\ell},
%  \math{f_{p}} and \math{F_{p}} are the density and distribution of
%  \math{p_{\ell,k}} and the last expression follows because
%  \math{q_{\ell}} is the maximum of \math{K} independent random variables
%  having density \math{f_{p}}. Performing the integral using
%  \math{f_p(s)=1/(b-a)} and \math{F_p(s)=(s-a)/(b-a)} gives
%  \math{\mu(K)} as claimed.
%\end{proof}
The {\sc max} process serves as an upper bound on group competence functions, as no other process can select a voter with a higher competence. Often the competence is not directly observable as in Example~\ref{example:uniform-max}. However, when the group size is not too large, for example a group is a household, then it is natural to assume that the family members are able to choose the max-informed representative. Nevertheless, this process might be noisy, which is captured in the next example.
%a natural assumption in many cases, especially when the group size is small. For example, when each group is a household, then its members often have good ideas about the competence of each other. The next example captures cases where such selection is noisy. %For other case, this assumption is generalized in the next example, where the representative is chosen with probability that is correlated with her competence.
\begin{ex}[Noisy-Max Process]\label{example:uniform-noisy}
\rm
As in Example~\ref{example:uniform-uniform}, suppose  $F=\text{Uniform}[a,b]$. Given a $K$-dimensional vector $\vec p=(p_1,\ldots,p_K)$ such that $p_1\,\cdots\ge p_K\ge  0$ and $\sum_{i=1}^K p_i=1$. For any $i\le K$, suppose the group chooses the member with $i$-th highest competence as the representative with probability $p_i$. Let $q_{\ell,i}$ denote the $i$-th highest competence. It follows that $\frac{q_{\ell,i}-a}{b-a}$ is a Beta random variable with PDF $\text{Beta}(K+1-i,i)$~\cite{Gentle2009:Computational}, whose mean is $\frac{K+1-i}{K+1}$. We call this representative democracy the {\sc noisy-max} process, denoted by $\mu_{\vec p}$. Therefore,
$$
\mu_{\vec p}(K)=\sum_{i=1}^K p_i \left(\frac{i}{K+1}a+\frac{K+1-i}{K+1}b\right)=(\vec p\cdot\vec k_+) a+(\vec p\cdot\vec k_-) b,
$$
where $\vec k_+=(\frac{1}{K+1},\frac{2}{K+1},\ldots,\frac{K}{K+1})$ and $\vec k_-=(\frac{K}{K+1},\frac{K-1}{K+1},\ldots,\frac{1}{K+1})$. \hfill$\blacksquare$
\end{ex}
It is easy to verify that the group competence function for {\sc max} process is monotonically increasing in $K$ and is between $a$ and $b$, while the group competence function for {\sc uniform} process outputs the same value for all $K$.

% !TEX root = wiser.tex
\section{Optimal Representative Democracy for One Issue}
In this section, we focus on the characterizing optimal representative democracy for one issue.
\subsection{Consistent Representative Democracy}
We first extend the classical Condorcet Jury Theorem to representative democracy.
Let us first formally define consistency, the main desired property of a
representative democracy.
As the number of voters increases, i.e. asymptotically
in \math{n}, it should be possible to choose a partition of $n$ voters into $L$ groups, with potentially different number of voters in each group, such that 
% \math{K} (perhaps as a function of \math{n}) so that the majority vote of the \math{L=n/K}
with probability 1 the majority representatives vote for $1$.

Given $\vec K=(K_1,\ldots,K_{L(n)})$, we let $S_{n,\vec K,\mu}$ denote the random variable that represents the fraction of $1$'s in $L(n)$ independent Bernoulli random variables with success probabilities $(\mu([\vec K(n)]_1),\ldots,$ $\mu([\vec K(n)]_{L(n)}))$, where for any $i\le L(n)$, $[\vec K(n)]_i=K_i$ is the $i$-th component of $\vec K(n)$. In other words, $S_{n,\vec K,\mu}$ is $\frac{1}{L(n)}$ of the Poisson trail that represents the representatives' votes. 

\begin{dfn}
Given a partition function $\vec K$ and a group competence function $\mu$, for any $n$, we let $R_n(\vec K,\mu)$ denote the probability for majority voting by representatives according to $\vec K$ and $\mu$ to succeed. That is, 
$R_n(\vec K,\mu)=\Prob(S_{n,\vec K,\mu}>\frac{1}{2})$.
\end{dfn}
The subscript $n$ in $R_n$ is sometimes omitted when causing no confusion.

%Given a selection function $\mu$ that is monotonically non-decreasing in $K$, and let there be $L\ge 1$ groups, each of which has $K$ members. Let $S_{K,L,\mu}$ denote the average value of random variable obtained by summing up $L$ Bernoulli trials, each of which takes $1$ with probability $\mu(K)$, and takes $0$ otherwise. Clearly majority voting of the representatives chooses the correct alternative ($1$) with probability $\Prob(S_{K,L,\mu}>\frac{1}{2})$. 

\begin{dfn}\label{def:rep-dem-consistent}
 A representative democracy with group competence function $\mu$ is {\em consistent} if there exists
  a partition function \math{\vec K(n)} for which the majority voting of representatives succeeds with probability $1$ as  \math{n\rightarrow\infty}, that is, $\lim_{n\rightarrow \infty} R_n(\vec K,\mu)=1$.
%  \mandc{
%    \lim_{n\rightarrow\infty}\Prob[S_{\vec K, L(n),\mu}>\frac{1}{2}]=1}
%  where \math{x_\ell\sim\text{Bernoulli}(\mu(K(n)))}
%  are i.i.d. Bernoulli random variables.
\end{dfn}
%The next theorem shows that any representative-democracy which satisfies  Axiom A is consistent with \math{K(n)=K_*},
%which justifies why we label Axiom A ``consistency''.
%We omit the proof which follows immediately from the law of large numbers.
\begin{thm}\label{thm:cjt}
For each distribution $F$, the representative democracy with group competence function $\mu$ is consistent if and only if there exists $K_*\in \mathbb N$ such that $\mu(K_*)>0.5$.
\end{thm}
\begin{proof} $\Rightarrow$:  suppose for the sake of contradiction, $\mu(K_*)\le 0.5$ but there exists $\vec K$ such that $\lim_{n\rightarrow\infty}R_n(\vec K,\mu)=1$. It follows that for any $n$, $R_n(\vec K,\mu)$ is no more than the average value of $L(n)$ Bernoulli trials, each succeeds with probability $0.5$. However, the probability for the latter to be strictly larger than $0.5$ is no more than $0.5$, because its PDF is symmetric and its mean is $0.5$. This leads to a contradiction .$\Leftarrow$:  we choose the partition function such that all but one group have $K_*$ members. It follows that as $n\rightarrow\infty$ the average competence is at least $\frac{n/K_*-1}{n/K_*+1}\mu(K_*)>\frac12$. This means that  $R_n(\vec K,\mu)=1$ as $n$ goes to infinity~\citep{Fey03:Note}.
\end{proof}

{\noindent\bf Costly voting.} We now formally define the benefit of correct decision, cost of voting, and the social welfare for optimization.

\begin{dfn} 
For any $n$ and $L$, let $\ben{n}\in\mathbb R$ denote the benefit of making the correct decision and let $\cost{L}$ denote the monetary cost of maintaining $L$ representatives. Given a partition $\vec K=(K_1,\ldots,K_{L(n)})\in {\mathbb N}^{L(n)}$, the social welfare of $\vec K$ is the expected benefit minus the cost of voting, that is,
$$\sw{\vec K}=\ben{n}\correct_n(\vec K,\mu)-\cost{L}=\ben{n}\left(\correct_n(\vec {K},\mu)-\frac{\cost{L}}{\ben{n}}\right)$$
\end{dfn}

In the following subsections we will characterize optimal $\vec K$ that maximizes $\sw{\vec K}$ for different cases.

%\begin{thm}\label{thm:consistent-constant}
%  If \math{\mu(K_*)\ge \frac12+\epsilon>\frac12},
%  then the representative-democracy is consistent if
%  \math{K(n_*)\ge K_*} for some \math{n_*} and \math{K(n)/n\rightarrow 0}.
%  In particular \math{K(n)=K_*} for
%  all \math{n\ge K_*} is consistent.  
%\end{thm}

\subsection{Optimal Group Size for Single Issue: Fixed Cost of Voting}

In this subsection we focus on the setting where the cost of voting is fixed regardless of the number of representatives. 
Since the cost of voting is fixed, the goal is to find the optimal partition $\vec K$ that maximizes $\correct_n(\vec K,\mu)$.

Before we formally present our  main results, let us first discuss the  effect of increasing the size of groups at a high level. The effect can be seen as a tradeoff between quality (competence of each representative) vs.~quantity (the total number of representatives), and we will see that it comes down to a form of mean vs.~variance tradeoff.

{\bf \noindent The quality vs.~quantity tradeoff.} Suppose $n$ is fixed and we are deciding to use group size of $K_1$ or $K_2$, where $K_1<K_2$ and $\mu(K_2)\ge \mu(K_1)>0.5$, and for the purpose of presentation, suppose that both divides $n$. That is, $n=K_1L_1=K_2L_2$. Let $\vec K_1=(\underbrace{K_1\ldots,K_1}_{L_1})$ and $\vec K_2=(\underbrace{K_2\ldots,K_2}_{L_2})$.

When there are $K_1$ members in each group, each representative has competence $\mu(K_1)<\mu(K_2)$. On the other hand, the number of representatives is $L_1>L_2$. Therefore, we have $\Exp(S_{n,\vec K_1,\mu})=\mu(K_1)\le \mu(K_2)=\Exp(S_{n,\vec K_2,\mu})$, while 
the variance of $S_{n,\vec K_2,\mu}$, which is  $\mu(K_1)(1-\mu(K_1))/L_1$ can potentially be smaller than the variance of $S_{n,\vec K_2,\mu}$.\footnote{It is possible for $S_{n,\vec K_2,\mu}$ to have a smaller variance, if $\mu(K)$ grows faster than linear.} Therefore, the optimal group size $K^*$ minimizes the left tail probability of $S<0.5$. This can roughly be seen as a mean vs.~variance tradeoff, as illustrated using the distribution of
$S_{n,\vec K,\mu}$
in Figure~\ref{fig:tradeoff}, where we set $n=2000$, $F=\text{Uniform}[0.44,0.55]$, the {\sc max} process $\mu_\text{max}$ is used, and we compare the distribution of three group sizes: $K=2$ (low mean, low variance), $K^*=8$ (optimal, middle), and $K=20$ (high mean, high variance).

\begin{figure}[htp]
\centering\includegraphics[trim=0cm 9.9cm 13.3cm 0, clip=true, width=.55\textwidth]{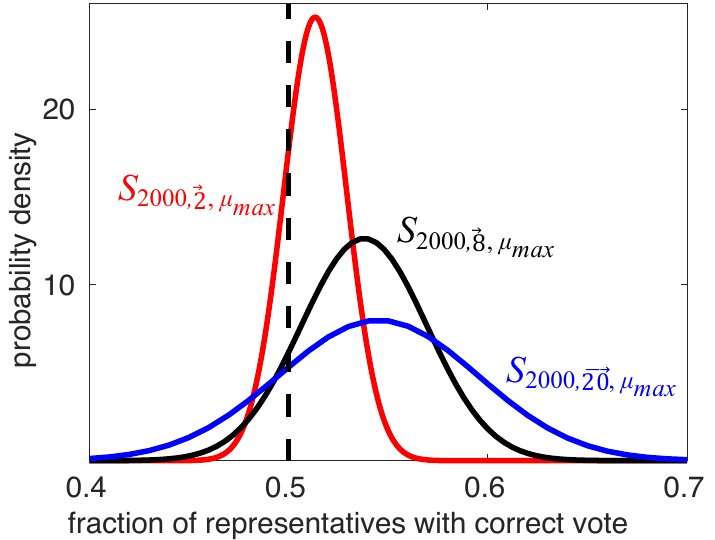}
\caption{The mean vs.~variance tradeoff in choosing different group sizes.\label{fig:tradeoff}}
\end{figure}

Naturally, our next goal is to identify an optimal number of representatives (groups), denoted by $L^*(n)$. 
%schedule for
%\math{K^*(n)}, the optimal group size as a function of
%\math{n} which results in the highest probability that the
%final vote results in the correct decision.
%Specifically, we are interested in the asymptotic behavior of
%\math{K^*(n)}. A constant number of representatives cannot be optimal,
%and so any optimal schedule for \math{K^*(n)} must be in \math{o(n)} so that
%\math{L^*(n)=n/K^*(n)} is in \math{\omega(1)}.
%Admittedly, there is a large range for \math{K^*} between
%\math{\Theta(1)} and \math{o(n)}, however this little progress does give
%us some leverage. Namely, since the number of
%representatives must grow asymptotically with
%\math{n}, we can use the Central Limit Theorem to analyze the majority
%vote of the representatives.
%  Let us now prove the main theorem of this section, namely
%  under the no-God hypothesis and if it is \emph{possible} to get
%  asymptotic consistency, then the optimal group size is a constant.
We will first focus on the specific partition functions where almost all groups have the same size, with the exception for the last group, whose size is allowed to be larger than other groups. This is a natural setting in practice because each representative is supposed to represent equal number of voters. Later in this section we will show how to extend our study to general partition functions. 
\begin{dfn}[Homogeneous groups]
  Given $n$ and a group size $K$, in the homogeneous setting,
  $L=\lfloor n/K\rfloor$ representatives are selected using
  partition function $\vec H_K=(\underbrace{K,\ldots,K}_{L-1}, n-(L-1)K)$. 
\end{dfn} 
%For homogenous groups of size \math{K}, all but one group has $K$ members, and this last group has at least $K$ members. 

Next, we prove a general
result: among all homogeneous group size settings, the optimal group size is   bounded by
    a constant,
    provided that some value of \math{K} can achieve consistency and the group competence function $\mu$ is polynomially bounded away from 1 as $K$ goes to infinity. %, that is, there exists $A<1$ and $\alpha\ge 0$ such that \math{\mu(K)\le 1-A/K^{\alpha}}. 
Let
\math{K_\text{hom}^*(n)} denote the optimal homogeneous group size that maximizes $\sw{\vec H_K}$, where $\cost{L}$ is a constant. In other words, \math{K_\text{hom}^*(n)} maximizes $\correct_n(\vec H_{K_\text{hom}^*(n)},\mu)$, the probability for
the majority of representatives to vote for the correct decision.

 \begin{thm}[Optimal homogeneous group size]\label{thm:constant} Suppose the $\cost{L}$ is a constant, there exists $K_*\in\mathbb N$ such that
      $\mu(K_*)\ge\frac12+\epsilon$ for \math{0<\epsilon<\frac12},
      and for all $K\in\mathbb N$, \math{\mu(K)\le 1-A/K^{\alpha}} for constants \math{A>0}
      (w.l.o.g. \math{A\le 1})
      and $\alpha \ge 0$. Then, \math{K_\text{hom}^*(n)\le c K_*}, where
      \mand{
        c=\frac{-4}{\ln(1-4\epsilon^2)}\left(        
    \ln \frac{32}{9\epsilon^2 A}+
    \alpha\ln\frac{4\alpha K_*}{e}-\alpha\ln|\ln(1-4\epsilon^2)|\right)
        .}
\end{thm}
    The intuition is that as \math{K} increases, we
    experience diminishing returns with respect to \math{\mu} because
    \math{\mu} is bounded away from 1. On the other hand, there is a loss due to decreasing \math{L=\floor{n/K}}, the number of
    representatives.  We find this constant bound surprising because one may expect that the best tradeoff is achieved when $K$ is a function of $n$, but our theorem proves that it is a constant.
 %This loss dominates because having fewer voters has an exponential impact. 
\begin{proof}(of Theorem~\ref{thm:constant}) We first prove a series of lemmas as building blocks.
  W.l.o.g.~suppose each of the first $L-1$ groups has exactly $K$ members.  
  We first observe an elementary bound that allows us to
  ignore the last group whose number of representatives is unknown,
  \mandc{
    \underbrace{\Prob\left[\sum\limits_{\ell=1}^{L-1}x_\ell\ge\ceil{\frac{L+1}{2}}\right]}_{\correct_-(L,\mu(K))}
    \le
    \correct_n(\vec H_{K},\mu)
    \le
    \underbrace{
      \Prob\left[\sum_{\ell=1}^{L-1}x_\ell\ge\ceil{\frac{L-1}{2}}\right]}_{\correct_+(L,\mu(K))}.
  }
  The functions \math{\correct_-} and \math{\correct_+} are Binomial upper tail probabilities,
\remove{We define a lower bound and an upper bound, where the competence of the last representative is $0$ and $1$, respectively.}
\begin{align*}
&\correct_-(L,p)=\sum_{\ell=\ceil{\frac{L+1}{2}}}^{L-1}\choose{L-1}{\ell}p^\ell(1-p)^{L-1-\ell}.\\
&\correct_+(L,p)=\sum_{\ell=\ceil{\frac{L-1}{2}}}^{L-1}\choose{L-1}{\ell}p^\ell(1-p)^{L-1-\ell}.
\end{align*}
For any $n,K,\mu$, and $L=\floor{n/K}$, we have $\correct_-(L,\mu(K))\le \correct_n(\vec H_{K},\mu)\le \correct_+(L,\mu(K))$. The following five lemmas are properties of $\correct_-(L,p)$ and $\correct_+(L,p)$. So far, we have not yet
invoked any properties of the group competence function $\mu$.
  \begin{lem}[Monotonicity of \math{\correct(L,p)}]\label{lem:binomial-monotonic}
    For fixed \math{L}, $\correct_-(L,p)$ and $\correct_+(L,p)$ are increasing in \math{p}. For fixed
    \math{p>\frac12}, $\correct_-(L,p)$ and $\correct_+(L,p)$ are increasing in \math{L}.
  \end{lem}

  \begin{lem}\label{lem:max-rplus}
    For \math{p\le\frac12}, \math{R_+(L,p)\le\frac34}
    (the maximum is attained for \math{L=3, p=\frac12}).
  \end{lem}

  We will prove bounds on $\correct_-(L,p)$ and $\correct_+(L,p)$, for which we will need bounds on
  binomial probabilities and central binomial coefficients. Specifically,
  \begin{lem}[Binomial Tail Inequality]\label{lem:binomial-tail}
    Given \math{p>\frac12}, \math{L} and \math{k\le \ceil{L/2}}, 
    \mand{\choose{L}{k}p^k(1-p)^{L-k}\le \sum_{\ell=0}^k\choose{L}{\ell}
      p^\ell(1-p)^{L-\ell}\le \frac{p}{2p-1}\choose{L}{k}p^k(1-p)^{L-k}.
      }
  \end{lem}
  \begin{proof}
    The lower bound is just the last term in the sum.
    The upper bound
    comes from the observation that each term in the sum is
    at least a factor \math{\mu/(1-\mu)} bigger than the previous term
    (because \math{k\le \ceil{L/2}}). %We omit the details
  \end{proof}
  
 \begin{lem}[Near-Central binomial coefficient bound]\label{lem:binomial-central} For \math{L>1},
    \mand{\frac{3}{4}\cdot\frac{4^{L/2}}{\sqrt{\pi L}}
      \le
      \choose{L}{\ceil{\frac12(L-1)}}\le 2\cdot \frac{4^{L/2}}{\sqrt{\pi L}}
    .}
  \end{lem}
 
 \begin{lem}[Bounding $\correct_-(L,p)$ and $\correct_+(L,p)$]\label{lem:correct-bounds}
      For \math{\frac12<p<1},
\mand{
        1-\left(\frac{2}{(2p-1)}\right)\cdot\frac{(4p(1-p))^{L/2}}{\sqrt{\pi L}} 
        \le
        \correct_-(L,p)
        \le
        \correct_+(L,p)
        \le
        1-\frac{3}{8p}\cdot\frac{(4p(1-p))^{L/2}}{\sqrt{\pi L}}.}
\end{lem}
  The proofs are relegated to the appendix.

We are ready to prove the theorem. Let
\math{c} be defined as in the statement of the theorem.
We may assume \math{n>cK_*} otherwise the theorem automatically holds.
Further, if \math{L*=1}, then there is just one group and any
\math{K>K_*} will also have just one group and be equivalent. Therefore, we
may assume \math{L_*\ge 2}.
  Now suppose \math{K>cK_*}.
  Define
  \math{\mu_K=\mu(K)} and \math{L_K=\floor{n/K}},
  \math{\mu_*=\mu(K_*)} and \math{L_*=\floor{n/K_*}}.
  Observe that \math{L_K\le L_*}.
  We show that \math{\correct(L_*,\mu_*)\ge\correct(
    L_K,\mu_K)}
  which means that \math{K} cannot be better than
  \math{K_*} for a homogeneous partition of \math{n},
  proving the theorem.

  If $\mu_K\le \frac12$ then \math{\correct_+(L_K,\mu_K)\le \frac34}
  (Lemma~\ref{lem:max-rplus}). We show that
  \math{\correct_-(L_*,\mu_*)>\frac34}.
  Indeed, since \math{n>cK_*}, we have \math{n/K_*>c\ge
    \frac{-4}{\ln(1-4\epsilon^2)}\cdot\ln \frac{32}{9\epsilon^2 A}
    }, and so
  \mandc{
    L_*=\floor{\frac{n}{K_*}}\ge
    \frac{n}{2K_*}> \frac{-2}{\ln(1-4\epsilon^2)}\cdot\ln \frac{32}{9\epsilon^2 A}
    \qquad\implies\qquad
    \frac{(1-4\epsilon^2)^{L_*/2}}{\epsilon\sqrt{\pi L_*}}
    <
    \frac{9\epsilon A}{32\sqrt{\pi L_*}}
    <\frac14,
  }
  where in the last inequality we used
  \math{A\le 1} and \math{L_*\ge 2}.
  Now, using the bound for \math{R_-} from Lemma~\ref{lem:correct-bounds},
  we conclude that \math{R_-(L_*,\mu_*)>\frac34}, which proves
  \math{K} cannot be optimal.
  Therefore, we may assume that \math{\mu_K>\frac12}.
\remove{  
  Let $N$ to be sufficiently large such that for all $n>N$, $\correct_-(L_*,\mu_*)>0.75$. If $\mu_K\le 0.5$, then $\correct_+(L_K,\mu_K)\le 0.75<\correct_-(L_*,\mu_*)$, which means that the theorem holds. Therefore, we may assume that $\mu_K>0.5$.}
  Also, if \math{\epsilon=\frac12} then the representatives always vote for the correct decision and the theorem automatically holds, so we may assume
  \math{\epsilon<\frac12}.
  Using Lemma~\ref{lem:correct-bounds},
  \eqan{
    &&\correct_n(\vec H_{K_*},\mu)-\correct_n(\vec H_{K},\mu)\ge \correct_-(L_*,\mu_*)-\correct_+(L_K,\mu_K)]\\[5pt]
    &\ge&
    \frac{3}{8}{\frac{1}{\mu_K}}\cdot\frac{(4\mu_K(1-\mu_K))^{L_K/2}}{\sqrt{\pi L_K}}
    -
    \left(\frac{2}{2\mu_*-1}\right)\cdot\frac{(4\mu_*(1-\mu_*))^{L_*/2}}{\sqrt{\pi L_*}}
    \\
    &{\buildrel {(L_*\ge L_K)}\over\ge}&
     \frac{3}{8}{\frac{1}{\mu_K}}\cdot\frac{(4\mu_K(1-\mu_K))^{L_K/2}}{\sqrt{\pi L_*}}
    -
    \left(\frac{2}{2\mu_*-1}\right)\cdot\frac{(4\mu_*(1-\mu_*))^{L_*/2}}{\sqrt{\pi L_*}}\\
    &\ge&
    (\text{positive})\cdot
    \left[1-
      C
      \left(\frac{(4\mu_*(1-\mu_*))^{L_*/L_K}}{4\mu_K(1-\mu_K)}\right)^{L_K/2}\right],
  }
  where \math{C=\frac{16}{3}\cdot \frac{\mu_K}{2\mu_*-1}
    =\frac{8\mu_K}{3\epsilon}}. Note that
  \math{C\le \frac{8}{3\epsilon}} (because \math{\mu_K<1})
  and \math{C>1} (because \math{\mu_K>\frac12}
  and \math{\epsilon<\frac12}).
  We prove the term in square parentheses is positive.
   Observe that
  \mand{
    \frac{L_*}{L_K}
    =\frac{\floor{n/K_*}}{\floor{n/K}}
    \ge\frac{\floor{n/K_*}}{{n/K}}
    =
    \frac{K}{K_*}\frac{\floor{n/K_*}}{n/K_*}
    \ge
    \frac{K}{2K_*}.
    }
We used \math{\floor{x}/x\ge\frac12} when
  \math{x\ge 1}. Because
  \math{\mu_K>\frac12} and \math{1-\mu_K\ge A/K^\alpha}, we have
  \math{\mu_K(1-\mu_K)\ge A/2K^\alpha}. Also recall that 
  \math{\mu_*\ge\frac12+\epsilon}, which means that $\mu_*(1-\mu_*)\le (\frac 12+\epsilon)(\frac 12-\epsilon)$, and \math{K\le n}. Therefore,
  \eqan{
    C
    \left(\frac{(4\mu_*(1-\mu_*))^{L_*/L_K}}{4\mu_K(1-\mu_K)}\right)^{L_K/2}
    &\le&
    C
    \left(\frac{K^\alpha(1-4\epsilon^2)^{K/2K_*}}{2A}\right)^{L_K/2}
      }
  We show that the RHS is at most 1, or equivalently its logarithm is at most zero,
  concluding the proof. Taking the logarithm of the RHS, we get:
  \eqan{
    &&
    L_K\left(\frac{K}{4K_*}\ln(1-4\epsilon^2)
    +
    \frac\alpha2\ln K
    -
    \frac12\ln 2A\right)+\ln C\\
    &{\buildrel {(L_K\ge 1,  C>1)}\over\le}&
     L_K\left(\frac{K}{4K_*}\ln(1-4\epsilon^2)
    +
   \frac\alpha2\ln K
    -
    \frac12\ln 2A+\ln C\right)\\
    &\le&
    L_K\left(\frac{K}{4K_*}\ln(1-4\epsilon^2)
    +
   \frac\alpha2\left(\ln\frac{-4\alpha K_*/e}{\ln(1-4\epsilon^2)}-
   \frac{\ln(1-4\epsilon^2)}{4\alpha K_*}K\right)
    -
    \frac12\ln 2A+\ln \frac{8}{3\epsilon}\right).
%\\    &=& 0 \hspace{50mm} \text{(Follows after the expression for \math{c}.)}
    }
The last step follows by using the  fact for
any \math{z>0}, \math{\ln x\le \ln(z/e)+x/z}, which holds because for any $y=x/z>0$, $\ln y-y+1$ is maximized at $y=1$. In the last step, we set
\math{z=-4\alpha K_*/\ln(1-4\epsilon^2)}.
Collecting terms,
we have
\eqan{
&&L_K\left(\frac{K}{8K_*}\ln(1-4\epsilon^2)
    +
   \frac\alpha2\ln\frac{-4\alpha K_*/e}{\ln(1-4\epsilon^2)}
    -
    \frac12\ln 2A+\ln \frac{8}{3\epsilon}\right)\\
    &=&
\frac{L_K   \ln(1-4\epsilon^2)}{8} \left(\frac{K}{K_*}-c
\right)\le0
}
where the last step follows because
\math{K/K_*>c} and \math{\ln(1-4\epsilon^2)<0}.
 \end{proof}
 
As a corollary, we will show that Theorem~\ref{thm:constant} can be applied to any $F$ with continuous density function, which means that in such cases the optimal number of groups with homogeneous size is $\Omega(n)$. To this end, we prove a lemma stating that any {\sc noisy-max} group competence function (Example~\ref{example:uniform-noisy}) is polynomially bounded away from 1.

\begin{lem}\label{lem:max-model-gen}
  For any {\sc noisy-max} group competence function $\mu_{\vec p}$ and any $F$ with a continuous density \math{f}
  on \math{[0,1]}, we have 
  \math{\mu_{\vec p}(K)\le 1-(\frac1{2\max(f)})/K}. Moreover, the max representative democracy $\mu_{\max}$ is concave, which means that \math{\mu_{\max}(K)+\mu_{\max}(K+2)\le 2\mu_{\max}(K+1)} for all \math{K\ge 1}.
\end{lem}
\begin{proof} It suffices to prove the inequality for $\mu_{\max}$ because expectation of any other order statistics is no more than it.
  Since \math{f} is continuous on the compact set \math{[0,1]}, it attains
  a maximum \math{B=\max(f)}. We note that since \math{f} is a density on \math{[0,1]}, \math{B\ge 1}. Let \math{F} be the CDF for \math{f}.
  We have
  \eqan{
    \mu_{\max}(K)&=&
    K\int_{0}^1 s\ f(S)F(s)^{K-1} ds=\int_{0}^1\ s (F(s)^{K})' ds=\left.s F(s)^{K}\right|_{0}^1-\int_{0}^1\ F(s)^{K} ds
  }
  Since \math{f(s)\le B} is continuous, \math{F(s)} is
  differentiable and 
  \math{F'(s)\le B}. Therefore, using  Taylor's theorem at \math{s=1},
  \math{F(s)=1-(1-s)F'(t)} for some \math{t\in[s,1]}, and since
  \math{F'(t)\le B}, we have
  \math{F(s)\ge \max(0,1-B(1-s))}. Therefore,
  \eqan{
    \mu_{\max}(K)&\le&
    1-\int_{0}^1ds\ \max(0,1-B(1-s))^{K}=1-\int_{1-1/B}^1ds\ (1-B(1-s))^{K}\\
    &=&1-\frac{1}{B(K+1)}\buildrel{(K+1\le 2K)}\over\le 1-({\textstyle\frac1{2B}})/K.
  }
Concavity of $\mu_{\max}$ holds because for any $K\ge 1$, we have 
\begin{align*}
\mu_{\max}(K)+\mu_{\max}(K+2)-2\mu_{\max}(K+1)=&-\int_{0}^1\ \left[F(s)^{K}+F(s)^{K+2}-2F(s)^{K+1} \right]ds\\
=&-\int_{0}^1\ F(s)^{K} (1-F(s))^2ds\le 0
\end{align*}
\end{proof}
Combining Theorem~\ref{thm:constant} and Lemma~\ref{lem:max-model-gen}, and notice that $\mu_{\max}$ provides an upper bound on the competence of the representative chosen by {\em any} representative selection process $\mu$, we have the following surprising corollary.

\begin{coro}[Constant group size]\label{cor:constant}
Suppose $F$ has a continuous density function on $[0,1]$. For any $\mu$ such that $\mu(K_*)>0.5$ for some $K_*$. The optimal number of representatives for homogeneous groups is at least $\frac{n}{cK_*}$, where $c$ is the constant defined in Theorem~\ref{thm:constant}.
\end{coro}
Therefore, the  representative democracy with fixed group size makes better choices than the  representative democracy with fixed number of representatives. More discussions can be found in the next subsection.
%to {\sc noisy-max group} competence function in the more general setting
%where  voters have a continuous
%density on \math{[0,1]}, even 
%without the no-God hypothesis. Using
%Lemma~\ref{lem:max-model-gen} with Theorem~\ref{thm:constant}, we have that the optimal group size to make the correct decision is growing at most as the logarithm of the number of
%voters.  That is quite small!

One limitation of Theorem~\ref{thm:constant} is that it only holds for homogeneous group size. 
Next, we will extend the theorem to inhomogeneous groups, allowing groups to have different sizes, at the cost of further requiring that the group competence function
\math{\mu(K)} to be concave.
%Using concavity and monotonicity of \math{\mu},
%we extend our main theorem to the general case
%where groups may be inhomogeneous (having different sizes).
%We show that the optimal number of groups must be large, near-linear in
%\math{n}. To do so,
Let $L^*(n)$ denote the optimal number of groups for $n$ voters.
\begin{thm}[Optimal number of representatives for general group sizes]\label{thm:oneissuegeneral} Suppose $\cost{L}$ is a constant, there exists $K_*\in\mathbb N$ such that
  $\mu(K_*)\ge\frac12+\epsilon$ for \math{0<\epsilon<\frac12}, $\mu$ is concave, and \math{\mu(K)\le 1-A/K^{\alpha}} for constant \math{A<1} and $\alpha \ge 0$. Then, for any \math{n}, $L^*(n)\ge \floor{\frac{n}{K_*}}/c$, where   
\mand{
        c=\frac{-4}{\ln(1-4\epsilon^2)}\left(        
    \ln \frac{32}{9\epsilon^2 A}+
    \alpha\ln\frac{4\alpha K_*}{e}-\alpha\ln|\ln(1-4\epsilon^2)|\right)
        .}
\end{thm}
\begin{proof}
  We show how to modify the proof for Theorem~\ref{thm:constant}. Let $R(\vec K_*,\mu)$ denote the probability that the majority
vote among the \math{L_*} representatives is correct under partition $\vec K_*$.
First, we may assume that
\math{\floor{n/K_*}>c}, because otherwise the
RHS in the bound is less than 1 and the theorem automatically holds.

Our key technique is to allow group size to be non-integers by extending $\mu$ to a piece-wise linear function $\bar \mu:{\mathbb R}_{\ge 0}\rightarrow [0,1]$. Then, we show that for any given number of groups $L$, $R(\vec K,\mu)$ is maximized when $\vec K$ partitions $n$ into $L$ equal-size groups, each with $\frac{n}{L}$ voters. This is guaranteed by concavity of $\mu$ and a result by~\citet{Fey03:Note}.

More precisely, let $\bar \mu:{\mathbb R}_{\ge 0}\rightarrow [0,1]$ be
the piecewise linear function that interpolates $\mu$ at the integers. It follows that $\bar\mu$ is concave and bounded above by $1-A/k^\alpha$ for all $k\ge 1$. For any $L\in \mathbb N$, we let $k=n/L$ and let $\bar R(L,p)$ denote the probability to obtain
a  majority of successes in $L$ Bernoulli trials, where each trial succeeds with probability $p$. The following Lemma gives an upper bound on the probability for majority voting to succeed, when the average competence is at most $0.5$. This answers an open question in~\cite[Lemma 5]{Owen89:Proving} and~\citep{Fey03:Note}.
\begin{lem}[Poisson-Binomial majority]\label{lem:poisson-binomial}
  A Poisson-Binomial random variable \math{X=x_1+\cdots+x_n} is a sum of
  \math{n}
  independent Bernoulli trials \math{x_i} 
  with respective (possibly different) success-probabilities
  \math{p_i}. Let \math{q} be the average probability,
  \math{nq=\sum_{i}p_i} and suppose
  \math{q\le\frac12}. Then, the probability to get a strict majority
  of successful trials is bounded by a constant. Specifically,
  \begin{enumerate}\itemsep-3pt\vspace*{-4pt}
  \item
    If \math{n=2K} is even: 
    \mldc{\Prob[X>n/2]\le1-B,\quad\text{ where }\quad
      B=\sum_{i=0}^{\floor{n/2}}\choose{n}{i}q^i(1-q)^{n-i}\le \frac 12\label{eq:maj-even}}
  \item
    If \math{n=2K+1} is odd and \math{n\ge1/(1-2q)} then
    \r{eq:maj-even} holds. Otherwise,
    \mandc{\Prob[X>n/2]\le\sqrt{1-B},\quad\text{ where }\quad
      B=\sum_{i=0}^{n}\choose{2n}{i}q^i(1-q)^{2n-i}\le \frac12\label{eq:maj-odd}}
      \end{enumerate}
\end{lem}
The bounds are in terms
of a standard Binomial with probability \math{q} equal to the average
probability in the ensemble. The bound for even
\math{n} is tight (set all probabilities to \math{q}).
The boundary case which we alluded to earlier is for odd
\math{n<1/(1-2q)}.
We conjecture that the bound for odd \math{n} can be improved
(for \math{q=\frac12}, we believe
\math{\Prob[\text{majority}]\le 1/\sqrt{e}},
which is better than \math{1/\sqrt{2}}).
It is interesting that the probability of a majority can
depend so much on whether the number of voters is even or odd,
and this is true even asymptotically (for \math{q=\frac12}, set
\math{K} of the \math{p_i} to \math{0} and \math{K+1} of the \math{p_i} to
\math{\frac{2K+1}{2K+2}}, then
\math{\Prob[\text{majority}]=(\frac{2K+1}{2K+2})^{K+1}\rightarrow 1/\sqrt{e}}). The proof is relegated to the appendix.
%The proof of Lemma~\ref{lem:poisson-binomial}
%is not crucial so we postpone it to
%after the main theorem.

Continuing with the proof of Theorem~\ref{thm:oneissuegeneral}, consider a partition $\vec K(n)$ with $|\vec K(n)|=L$ and define the average 
competence \math{q} of the representatives by $Lq=\sum_{i=1}^L\mu([K(n)]_i)$.
If \math{q\le\frac12}, then by
Lemma~\ref{lem:poisson-binomial},
\math{R(\vec K,\mu)<\frac{1}{\sqrt 2}}.
Using the bound for \math{\correct_-} in
Lemma~\ref{lem:correct-bounds}, with 
\math{\floor{n/K_*}>\frac{-4}{\ln(1-4\epsilon^2)}(\frac\alpha2\ln n
    +\ln \frac{4\sqrt{2}}{3\epsilon\sqrt A})}, we have 
$R(\vec K_*,\mu)>\frac{1}{\sqrt 2}$ and so \math{\vec K(n)} cannot be
optimal. Therefore, we may assume that \math{q>\frac12}.

By concavity of $\bar \mu$, for any $\vec K(n)$ with
$|\vec K(n)|=L$, we have $\bar\mu(k)=\bar\mu(n/L)\ge \sum_{i=1}^L\frac1{L}\mu([K(n)]_i)=q>\frac12.$
By~\citep{Fey03:Note}, since \math{q>\frac12}, the probability of a majority is
maximized when all the probabilities are the same as the average value \math{q}. And by
monotonicity, $\bar R(L,\cdot)$ can only increase in going from
\math{q} to \math{\bar\mu(k)}.
Therefore, $\bar\correct(L,\bar\mu(k))\ge \correct(\vec K,\mu)$.
What we have established is that if you fix the number of
groups, then we may assume that the probability to get a majority cannot
be better than the outcome where all group sizes are
homogeneous (possibly non-integer), equal to \math{n/L}. This gives us an upper
bound on the quality of the number of groups \math{L}.

The rest of the proof
is to show that if \math{L} is smaller than the claim of the theorem,
then this upper bound on the quality is less than the choosing
\math{L_*} equal to the bound in the Theorem, therefore
proving that \math{L} is not optimal.
This part of the proof
proceeds in the same way as the proof of Theorem~\ref{thm:constant}, by replacing $K$ by $k$ and replacing $\mu(K)$ by $\bar \mu(k)$. Lemma~\ref{lem:binomial-monotonic}--\ref{lem:correct-bounds} still hold because they do not depend on $k$. %For any $L<L_*/c$, we prove $R(\vec K_*,\mu)\ge R(L_*,\mu) \ge \bar R(L,\bar\mu(k))$. 
We note that the key steps involving $K$ in the proof of Theorem~\ref{thm:constant} are $\mu(K)\ge 0.5$ and  $1-\mu(K)\ge A/K^\alpha$, which hold for $k$ because we focus on $L$ such that $\mu(n/L)>0.5$ and $\bar\mu(k)\le 1-A/k^\alpha$.
\end{proof}

%{\bf Remarks on Theorem~\ref{thm:constant} vs.~Theorem~\ref{thm:oneissuegeneral}.}  Theorem~\ref{thm:constant} applies to homogeneous group sizes while Theorem~\ref{thm:oneissuegeneral} holds for arbitrary partition function. On the other hand, Theorem~\ref{thm:oneissuegeneral}  requires that $\mu$ to be concave while Theorem~\ref{thm:constant} does not. Both theorems require that $\mu$ approaches to $1$ slower than some polynomial, which is a very mild assumption because it holds for a wide range of representative selection processes, see Lemma~\ref{lem:max-model-gen} and the discussion afterwards.

Next, we further extend Theorem~\ref{thm:oneissuegeneral} by showing that the optimal partitioning function $\vec K$ must be 
nearly homogeneous. This property comes at the cost of requiring that
\math{\mu} is 
log-concave and \math{1-\mu} is log-convex. Log-concavity is
a weaker assumption than concavity. %however we also assume log-convexity of
%\math{1-\mu}. 
Let us consider an example.
\begin{ex}[Log-concavity and log-convexity of $\mu_{\max}$]\rm
  For uniformly distributed voters $F=\text{Uniform}[0,1]$ as in Example~\ref{example:uniform-uniform},
  the {\sc max} group competence function \math{\mu_{\max}(K)=K/(K+1)} is log-concave (it is also
  concave). %Further, \math{1-\mu_{\max}} is log-convex. We already showed concavity.
  Log-convexity of \math{1-\mu_{\max}} holds because
  \math{1-\mu_{\max}(K)=1/(K+1)} and \math{(K+2)^2>(K+1)(K+3)}, therefore
  \math{(1-\mu_{\max}(K))(1-\mu_{\max}(K+2))>(1-\mu_{\max}(K+1))^2}. 
  \hfill$\blacksquare$
  \end{ex}
\begin{thm}[Near-Homogeneity of group sizes]\label{thm:groups-homogeneous}
  Suppose $\cost{L}$ is a constant, the group competence function \math{\mu(K)} is log-concave
  and non-decreasing, and further that
  \math{1-\mu(K)} is log-convex. Given \math{L} groups, there is an optimal
  partition \math{(K_1,\ldots,K_L)} of \math{n} into the \math{L} groups
  with no two groups differing in
  size by more than 1. That is, \math{\max_{i}K_i-\min_iK_i\le 1}.
\end{thm}
\begin{proof}
  Suppose we have a partition of \math{n} voters
  into \math{L} groups with sizes
  \math{K_1\le K_2\le\cdots\le K_L}, and suppose that \math{K_L-K_1\ge 2}.
  Let \math{\mu_1=\mu(K_1)} and \math{\mu_L=\mu(K_L)}.
  Let the groups \math{K_2,\ldots,K_{L-1}} yield an arbitrary ensemble
  \math{A}
  of \math{L-2} representatives which is fixed in this proof.
  Define the functions
  \eqan{
    f(k)&=&\Prob[\text{\math{k} successes in \math{A}}]\\
    Q(k)&=&\Prob[\text{at least \math{k} successes in \math{A}}]=
    f(k)+f(k+1)+\cdots+f(L-2).
  }
  Let \math{M=\ceil{(L+1)/2}} be the majoirty threshold for the \math{L}
  representatives and define \math{P} as the probability that a majority of
  successes is obtained from all the representatives. Conditioning on the
  votes of representatives \math{1} and \math{L}, we have,
  \eqan{
    P&=&\mu_1\mu_L Q(M-2)+(\mu_1(1-\mu_L)+\mu_L(1-\mu_1))Q(M-1)+(1-\mu_1)(1-\mu_L)
    Q(M)\\
    &=&
    \mu_1\mu_L(Q(M-2)+Q(M)-2Q(M-1))+(\mu_1+\mu_L)(Q(M-1)-Q(M))+Q(M)\\
    &=&
    \mu_1\mu_L (f(M-2)-f(M-1))+(\mu_1+\mu_L)f(M-1)+Q(M)\\
    &=&
    \mu_1\mu_L f(M-2)+(\mu_1+\mu_L-\mu_1\mu_L)f(M-1)+Q(M)\\
    &=&
    \mu_1\mu_L f(M-2)+(1-(1-\mu_1)(1-\mu_L))f(M-1)+Q(M)
  }
  We now consider the partition of the \math{n} voters obtained by keeping the
  ensemble \math{A} fixed, increasing \math{K_1} to \math{K_1+1} and decreasing
  \math{K_L} to \math{K_L-1}. Let \math{P'} be the probability of a majority
  for the new representatives. The only difference with
  \math{P} is that \math{\mu_1\rightarrow \mu(K_1+1)} and
  \math{\mu_L\rightarrow\mu(K_L-1)}. By log-concavity of \math{\mu},
  \math{\mu(K_1+1)\mu(K_L-1)\ge \mu_1\mu_L}. By log-convexity of
  \math{1-\mu},
  \math{(1-\mu(K_1+1))(1-\mu(K_L-1))\le (1-\mu_1)(1-\mu_L)}. Therefore
  \math{P'\ge P}.
  As long as the difference in largest and smallest sizes is at least 2,
  we can continue merging without decreasing \math{P}, proving the theorem.
\end{proof}

Theorem~\ref{thm:groups-homogeneous} also applies to non-constant cost functions because the number of groups is fixed.

\subsubsection{Numerical Example of Optimal Group Size}
As an application of our results, we consider
the uniform voters with {\sc max} representative selection process (Example~\ref{example:uniform-max}) applied to US House.
Below, we show how the
optimal homogeneous group size \math{K^*} and the minimum group size required
to achieve consistency \math{K_*} depend on the upper bound on a voter's
probability of being correct,~\math{b}. We fix \math{a}, the lower bound,
to \math{0.45}, so \math{F=\text{Uniform}[0.45,b]}.
%\begin{center}
%  \resizebox{0.4\textwidth}{!}{\includegraphics*{UniformMaxOptK.eps}}
%\end{center}
\begin{figure}[htp]
\centering\includegraphics[trim=0cm 0cm 0cm 0, clip=true, width=.4\textwidth]{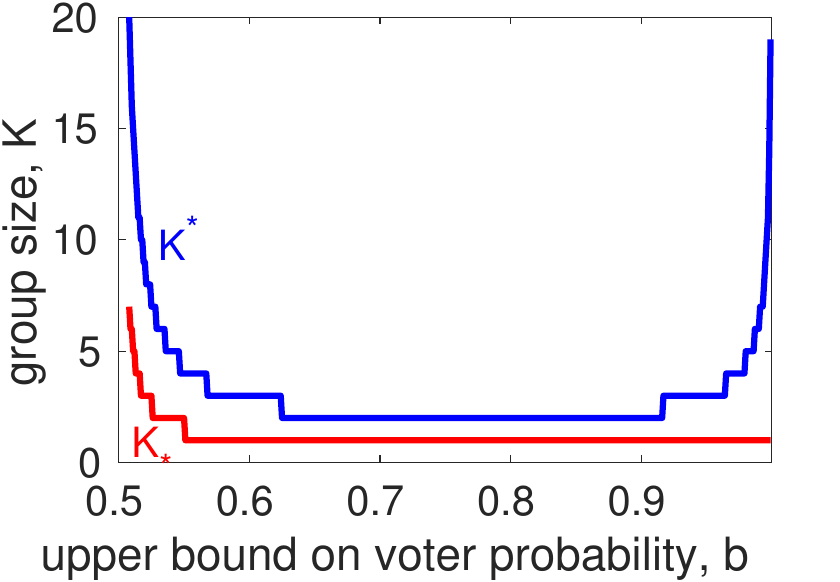}
\caption{Optimal group size.}
\end{figure}
When $b$ is small, as voters get wiser ($b$ goes up), \math{K_*} and \math{K^*} are decreasing.
At some point, even the direct-democracy (\math{K_*=1}) is consistent.
For very large \math{b}, the optimal group size starts
to increase due to the \math{\mu(1-\mu)} term.
In general, the optimal group size is less than about 5, the size of the
typical household. The optimal representative democracy is obtained when
each household elects a head to vote on its behalf.

The optimal group size is useful to know, but
for practical purposes, there may not be a significant difference between
different values of \math{K} for a large \math{n}. Let us take the
US House as an example, which has 435 representatives.
Suppose the voting population is about \math{n=\text{235 million}},
and that voter
competence is uniformly distributed from 0.45 to 0.52 (the average competence
is slightly less than \math{0.5}).
\begin{center}
  {\tabcolsep5.5pt
    \begin{tabular}{c|c|c|c|c|c|c}
    & \math{K=1} &\math{K_*=3}&\math{K^*=9}
  &\math{5\times\text{House}}
  &\math{2\times\text{House}}
  &\math{\text{House}}\\\hline
  Success rate &0\%&
  100\%&100\%&97\%&88\%&80\%
\end{tabular}
  }
  \end{center}
  In this simple setting, direct democracy (\math{K=1})
  would be wrong, and the current size of House is far from
  optimal, 20\% less accurate than what is achievable.
  Doubling congress gets you to 88\% and multiplying by 5 pretty much gets
  you to optimal. A House that is 20 times larger (each member representing 35K citizen) would be essentially
  indistinguishable from optimal. This  suggests that a much larger House is needed for noisy issues like the one in this example. Also note that if group sizes are about 5 (the size of a household) then we have near-perfection results.

%\subsection{Optimal Group Size for Single Issue: Exponentially Small Cost of Voting}

\subsection{Optimal Group Size for Single Issue: Polynomial Cost of Voting}
In this subsection, we focus on the setting where the cost of voting and the benefit of correct decision are both polynomial. 
%We start with characterizing the optimal group size $K_\text{hom}^*(n)=\arg\max_K \sw{\vec H_K}$.
  
 \begin{thm}[Optimal homogeneous group size, polynomial cost and polynomial benefit]\label{thm:constantlinearcost} Suppose $\frac{\cost{L}}{\ben{n}}=\Theta(\frac{L^{q_1}}{n^{q_2}})$ for constants $q_1>0$ and $q_2>0$, there exists $K_*\in\mathbb N$ such that
      $\mu(K_*)>\frac12$, $\mu$ is non-decreasing,
      and for all $K\in\mathbb N$, \math{\mu(K)\le 1-A/K^{\alpha}} for constants \math{A>0}
            and $\alpha \ge 0$. Then, the optimal group size \math{K_\text{hom}^*(n)=\Omega(n/\log n)}.

Moreover, we have:

(i) If $\lim_{K\ra\infty}\mu(K)<1$, then   \math{K_\text{hom}^*(n)=\Theta(n/\log n)}.

(ii) If  there exists $B,\beta>0$ such that $\mu(K)\ge 1-B/K^\beta$, then \math{K_\text{hom}^*(n)=\Theta(n)}.

            \end{thm}
\begin{proof} We first prove that the optimal number of groups $L^*_\text{Hom}$ is $O(\log n)$. 
Let $K'$ be the smallest number such that $\mu(K')>0.5$.  For any $K>K'$ and $L=\lfloor n/K \rfloor$, it follows from Lemma~\ref{lem:correct-bounds} that
$$\correct_n(\vec H_{K},\mu)-\frac{\cost{L}}{\ben{n}}=1-Q(L,n),$$
where $Q(L,n)=\Theta(\frac{(4\mu(K)(1-\mu(K)))^{L/2}}{\sqrt{ L}}+\frac{L^{q_1}}{n^{q_2}})$. This is  because $\mu(K)>\mu(K')>0.5$.

%Therefore, the optimal number of groups $L_\text{Hom}^*$ maximizes $\correct_n(\vec H_{K^*_\text{Hom}},\mu)-\frac{\cost{L_\text{Hom}^*}}{\ben{n}}$, where $L_\text{Hom}^*=\lfloor n/K_\text{Hom}^*\rfloor$. It follows from Lemma~\ref{lem:correct-bounds} that
%$$\correct_n(\vec H_{K^*_\text{Hom}},\mu)-\frac{\cost{L_\text{Hom}^*}}{\ben{n}}=1-Q(L_\text{Hom}^*,n)$$
%where $Q(L_\text{Hom}^*,n)=\Theta(\frac{(4\mu(K^*_\text{Hom})(1-\mu(K^*_\text{Hom})))^{L^*_\text{Hom}/2}}{\sqrt{ L^*_\text{Hom}}}+\frac{L^*_\text{Hom}}{n})$. This is  because $\mu(K_\text{Hom}^*)>0.5$ and $\mu$ is non-decreasing.

Therefore, $L^*_\text{Hom}=\arg\min_{L}Q(L,n)$. For any $L\ge n/K'$, we let $p(L)=4\mu(\lceil n/L \rceil )(1-\mu(\lceil n/L\rceil))$, which means that $p(L)$ is non-decreasing in $L$. Let $c_1,c_2>0$ be such that for all $L\ge n/K'$,
$$c_1(\frac{p(L)^{L/2}}{\sqrt{L}}+\frac{L^{q_1}}{n^{q_2}})\le Q(L,n)\le c_2(\frac{p(L)^{L/2}}{\sqrt{L}}+\frac{L^{q_1}}{n^{q_2}})$$
Let $p'=4\mu(K')(1-\mu(K'))$. We have $p'<1$. Let $c=-2q_2/\log p'$, which means that $(p')^{\frac c2 \log n}=\frac{1}{n^{q_2}}$. When $n\ge2^{1/c}$, we have
\begin{align*}
Q(c\log n, n)&\le c_2(\frac{(p')^{\frac c2\log n}}{\sqrt {c\log n}}+\frac{(c\log n)^{q_1}}{n^{q_2}})\le c_2(\frac{1}{n^{q_2}\sqrt {c\log n}}+\frac{(c\log n)^{q_1}}{n^{q_2}})\le 2c_2\frac{(c\log n)^{q_1}}{n^{q_2}}
\end{align*}
For any $L\ge c \sqrt[q_1]{\frac{2c_2}{c_1}} \log n$, we have $Q(L, n)>c_1L^{q_1}/n^{q_2}=2c_2\frac{(c\log n)^{q_1}}{n^{q_2}} \ge Q(c \log n , n)$. It follows that $\arg\min_LQ(L, n)\le \frac{c_2c}{c_1}\log n =O(\log n)$.

To prove (i), let $b=\lim_{K\ra\infty}\mu(K)<1$. %Let $c' = $ and let $n$ be large enough such that $\mu(\frac{n}{c'\log n})\ge \frac b2$. 
For any $L<-\frac{q_2}{\log (4b(1-b))}\log n$, we have 
\begin{align*}
&Q(L,n)> c_1\frac{p(L)^{L/2}}{\sqrt{L}}
>c_1\frac{(4b(1-b))^{L/2}}{\sqrt{L}}>c_1\frac{(4b(1-b))^{-\frac{q_2}{2\log (4b(1-b))}\log n}}{\sqrt{L}}=\Theta(\frac{1}{n^{q_2/2}\sqrt {\log n}})\\
\end{align*}
It follows that when $n$ is large enough, for any $L<-\frac{q_2}{\log (4b(1-b))}\log n$, we have  $Q(L,n)> 2c_2\frac{(c\log n)^{q_1}}{n^{q_2}}=Q(c\log n, n)$. This means that  $K_\text{hom}^*(n)=\Omega(n/\log n)$.

To prove (ii), let $\bar L=2q_2/\beta$, it is not hard to check that $Q(\bar L,n)=O((\frac{4B}{(n/\bar L)^\beta})^{\bar L/2}/\sqrt {\bar L} +\bar L^{q_1}/n^{q_2})=O(1/n^{q_2})$. Therefore, for any $L$ larger than some constant, we have $Q(L,n)>L^{q_1}/n^{q_2}>Q(\bar L,n)$, which means that $K_\text{hom}^*(n)=\Theta(n)$.
\end{proof}

\section{Optimal Representative Democracy for Multiple Issues}
\label{sec:multi}
We now extend our setting to  \math{d\ge 2} issues, \math{I_1,\ldots,I_d} with fixed cost of voting. Each voter's type is now represented by a  $2^d$-dimensional {\em competence vector} $\vec c$, which is a distribution over $\{0,1\}^d$. A voter's competence vector  represents her probability to cast a combination of votes over the $d$ issues. Let $\Delta_{2^d}$ denote the \math{2^d}-dimensional simplex, which is the set of all possible competence vectors. We assume that each voter's competence vector is generated i.i.d.~from a distribution $F$ over $\Delta_{2^d}$.

\begin{ex}\label{ex:multi}\rm
For two issues $I_1$ and $I_2$, let $F$ be the uniform distribution over two competence vectors $\vec c_1,\vec c_2$, as illustrated below, where
$(0,1)$ means that the voter is wrong on issue $1$ but correct on issue~$2$.

%\begin{table}[htp]
%\caption{$F$ is the uniform distribution over two competence vectors.\label{tab:exmulti}}
%\centering
\begin{center}
{\renewcommand{\arraystretch}{1.2}
\begin{tabular}{|c|c|c|c|c|}
\hline
$F$&$(0,0)$&$(0,1)$&$(1,0)$&$(1,1)$\\
%\hline
%$\frac{1}{3}@\vec c_1$& $\frac14$& $\frac14$& $\frac14$& $\frac14$ \\
\hline 
$\frac{1}{2}@\vec c_1$& $2/9$& $4/9$& $1/9$& $2/9$ \\
\hline
$\frac{1}{2}@\vec c_2$& $1/9$& $1/9$& $4/9$& $3/9$ \\
\hline 
\end{tabular}}
\end{center}
%\end{table}

The two issues are independent in $\vec c_1$: issue $I_1$ (respectively, $I_2$) takes $1$ with probability $\frac13$ (respectively, $\frac 23$), independent of the other issue. The two issues are correlated in $\vec c_2$.\hfill$\blacksquare$\end{ex}

Suppose \math{K} voters
\math{v_{\ell,1},\ldots,v_{\ell,K}} 
choose a representative \math{r_\ell} with competence vector $\vec c_\ell$.
As with one issue, the
representative can be summarized by $\Exp(\vec c_\ell)\in \Delta_{2^d}$. Thus,
we can succinctly describe a
\math{d}-issue representative selection process as a
group competence function.
\begin{dfn}\label{def:rep-dem}
  A representative selection process over \math{d} issues is a {\em group competence function} \math{\rho:\N\mapsto \Delta_{2^d}}.\footnote{We use $\rho$ to distinguish from single-issue group competence function $\mu$.} The multi-issue
  competence function
  \math{\rho} induces \math{d} single-issue {\em marginal group competence functions}
  \math{\mu_1, \mu_2, \ldots, \mu_d}, where \math{\mu_i} is  the marginal distribution for
  issue \math{I_i}. For any $K\in \mathbb N$, let \math{\rho(K)=(p_{\alpha_1,\alpha_2,\ldots,\alpha_d}:\forall i, \alpha_i\in\{0,1\})}.  Then,
$
  \mu_i(K)=\sum_{\bm\alpha:\alpha_i=1}p_{\alpha_1,\alpha_2,\ldots,\alpha_d}$.
\end{dfn}

\begin{ex}\label{ex:multirep} \rm
  Continuing Example~\ref{ex:multi}, suppose for a group of two voters
  independently sampled from \math{F},
  the representative selection process is to choose the voter with maximum expected number of correct votes. Then, the representative will be a $\vec c_2$
  voter with probability $\frac 34$ and a $\vec c_1$ voter with probability $\frac 14$. This group competence function, called {\sc max-sum} and denoted by $\rho_\text{ms}$, will output 
$$\rho_\text{ms}(2)=\frac 14\vec c_1+\frac34 \vec c_2=\left(\frac{5}{36},\frac{7}{36},\frac{13}{36},\frac{11}{36}\right)$$
And the marginal group competence functions are $\mu_1(2)=\frac{13}{36}+\frac{11}{36}=\frac23$, $\mu_2(2)=\frac{7}{36}+\frac{11}{36}=\frac12$.\hfill$\blacksquare$
\end{ex}

As with a single issue, when there are $d$ issues, given a partition function $\vec K$ and a group competence function $\rho$, we let $S^d_{n,\vec K,\rho}$ denote the $2^d$-dimensional random variable that is the average of $L(n)=|\vec K(n)|$ independent random variables $(\rho([\vec K(n)]_1),\ldots,\rho([\vec K(n)]_{L(n)}))$, where each $\rho([\vec K(n)]_i)$ represent the random vote on $d$ issues by the representative of group $i$. We let $\correct^d_n(\vec K,\rho)=\Prob[S^d_{n,\vec K,\rho}>\frac 12\cdot \vec 1]$ denote the probability that majority voting is correct on all $d$ issues.

\subsection{Consistent Representative Democracy for Multiple Issues}

Our next theorem extends the Condorcet Jury Theorem to representative democracy with multiple issues. It states that when the marginal group competence functions $\mu_i$'s are monotonic, the representative democracy is consistent if and only for each issue, there exists a group size for which marginal group competence of that issue is strictly larger than $0.5$.
\begin{thm}\label{thm:consistent-multiple}
  For $d\ge 2$, let $\rho$ be a $d$-issue group competence function
  with monotonic marginals $\mu_i$.\footnote{Technically, all we need is that
    \math{\mu_i(K)} can be lower-bounded by an increasing function of \math{K},
    and that this lower bound exceeds \math{0.5} for some \math{K_i}.}
Then, $\rho$ is consistent w.r.t. every
  issue if and only if every marginal is consistent, i.e.
for all $i\in\{1,\ldots,d\}$, there exists $K_i$
  such that $\mu_i(K_i)>\frac 12$.
%suppose 
%  If the induced single-issue representative democracies
%  are monotonic and consistent, then the multiple-issue
%  democracy is consistent.
%  In particular, set \math{K(n)=K_*=\max_i K_i},
%  then \math{\mu(K_*)} is consistent. Specifically,
%  \mand{\Prob[\text{representatives get all issues correct}]\ge
%    1-de^{-2n\min_i\epsilon_i^2/K_*}.}
\end{thm}
The proof uses the union bound and is similar to the proof of Theorem~\ref{thm:cjt}. The full proof can be found in the appendix.

%So we know that a multi-issue representative democracy is consistent
%if the marginals are consistent and monotonic. 
%The rate of convergence to consistency is governed by
%the slowest convergence over the \math{d}
%issues, therefore the choice of \math{K} should be picked
%optimally for that slowest issue using the methods from the previous section.
\subsection{Optimal Group Size for Multiple Issues}
We now
prove the analog of Theorem~\ref{thm:constant} for multiple issues. Namely
an upper bound on  the optimal homogeneous group size for any consistent multi-issue
representative democracy. We recall that $K_\text{hom}^*(n)$ is the optimal partition function for $n$ voters.

\begin{thm}\label{thm:multiissue} Let \math{\rho} be a $d$-issue group
  competence function for which there exists $K_*\in\mathbb N$ such that
for all $i\le d$, $\mu_i(K_*)\ge\frac12+\epsilon$ for constant \math{0<\epsilon<\frac12},
      and \math{\mu_i(K)\le 1-A/K^{\alpha}} for constant \math{A<1} and $\alpha \ge 0$. Then, $K_\text{hom}^*(n)\le (c- \frac{8}{\ln(1-4\epsilon^2)} \ln d) K_*$, where 
      \mand{
        c=\frac{-4}{\ln(1-4\epsilon^2)}\left(        
    \ln \frac{32}{9\epsilon^2 A}+
    \alpha\ln\frac{4\alpha K_*}{e}-\alpha\ln|\ln(1-4\epsilon^2)|\right)
        .}      
%      \math{K_\text{hom}^*(n)\le cK_*}, where 
%      \mand{
%        c=\max\left\{\frac{\ln 2d}{\epsilon^2},
%        \frac{-4}{\ln(1-4\epsilon^2)}\left(\alpha\ln n+\ln d
%        +\ln\frac{\sqrt{2}(1+2\epsilon)}{3A\epsilon}\right)\right\}
%        .}
\end{thm}
\begin{proof}
  Given any distribution $\gamma$ over $\{0,1\}^d$ and any $L\ge 1$, we let $Q_-(L,\gamma)$ (respectively, $Q_+(L,\gamma)$) denote the probability that for
  every one of the $d$ issues, the majority of voters vote for $1$, where there are $L$ independent voters, the first $L-1$ vote according to $\gamma$, and the last voter always vote for $0$ (respectively $1$) for all issues. We first extend Lemma~\ref{lem:correct-bounds} to multiple issues.

%For a multi-issue representative democracy,
%define \math{\correctM(L,\mu)} to be the probability that the
%representative vote is correct for \emph{all} the issues, where
%\math{L=\floor{n/K}} and \math{\mu} is the distribution over the vote-vector
%\math{\x_\ell} for representative \math{r_\ell}.
%As usual, let \math{\mu_i(K)} be the marginal probabilities for each
%issue \math{\nu_i} and let \math{K_*} be such that
%\math{\mu_i(K_*)\ge\frac12+\epsilon} for each marginal
%\math{i=1,\ldots,d}. We are interested in \math{K^*(n)}, the optimal
%group size which maximizes \math{\correctM(L^*,\mu^*)}.
%We still define the probability to get a single
%issue correct using \math{\mu_i} by \r{eq:correct}.
%First we give bounds on \math{\correct(L,\mu)}.
\begin{lem}\label{lem:bound-correctM}
  For any multi-issue representative democracy with group competence function \math{\rho} and
  given
  \math{n, K}, let the marginals be \math{\mu_i} and let
  \math{\mumin=\min_i\mu_i(K)} be the minimum marginal. Then,
  \mand{
    1-d\left(\frac{2}{(2\mumin-1)}\right)\cdot\frac{(4\mumin(1-\mumin))^{L/2}}{\sqrt{\pi L}}
    \le\correctM_-(L,\rho(K))\le\correctM_+(L,\rho(K))\le
      1-\frac{3}{8\mumin}\cdot\frac{(4\mumin(1-\mumin))^{L/2}}{\sqrt{\pi L}}
  }
\end{lem}
We can now mimic the analysis in the proof of Theorem~\ref{thm:constant}
to prove an upper bound on the optimal group size in a multi-issue
representative democracy.
%  Suppose \math{K\ge cK_*}.
%  As in the proof of Theorem~\ref{thm:constant},
%  we observe that \math{L_*/L\ge c/2} which means \math{L_*\ge c/2} since
%  \math{L\ge 1}. Since \math{K\ge(K_*/\epsilon^2)\ln 2d}, this
%  ensures that \math{\mumin\ge \frac12}.
%  Indeed, if \math{\mumin<\frac12} then
%  \math{\correctM(L,\mu)\le\correct(L,q)< \frac12}.
%  But, from the proof of Theorem~\ref{thm:consistent-multiple},
%  \mand{
%    \correctM(L_*,\mu_*)\ge 1-de^{-2L_*\epsilon^2}
%    \ge 1-de^{-c\epsilon^2}=1-de^{-\ln 2d}=\frac12.}
%  Therefore \math{\correctM(L_*,\mu_*)\ge \correctM(L,\mu)} and so
%  we may assume that \math{q\ge \frac12}.
%
%  The rest of the
%  proof mimics the proof of Theorem~\ref{thm:constant}, so we
We  give only the main steps, omitting some details. Let $L_*=\floor{n/K_*}$, $q_*=\min_i\mu_i(K_*)$, $L_K=\floor{n/K}$, and $q_K=\min_i\mu_i(K)$.
  Using Lemma~\ref{lem:bound-correctM},
  \math{q_*\ge\frac12+\epsilon} and \math{q_k\le 1-A/K^\alpha}
  (which implies \math{q(1-q)\ge A/2K^\alpha}), we
  get
\begin{align*}
\correct^d_n(\vec H_{K_*},\rho)-\correct^d_n(\vec H_{K},\rho)
    \ge &\correctM_-(L_*,\rho(K_*))-\correctM_+(L_K,\rho(K))  
     \ge 
    (\text{pos})
        \left[1-
      dC\left(\frac{K^\alpha(1-4\epsilon^2)^{K/2K_*}}{2A}\right)^{L_K/2}\right]
\end{align*}
Here $C={8\mu_K}/{3\epsilon} \le {8}/{3\epsilon}$.
  It now remains to prove that the expression in square parentheses
  is positive, for which it suffices to show that the logarithm of the
  second term is at most 0. Taking the logarithm of the second term,
   \eqan{
    &&
    L_K\left(\frac{K}{4K_*}\ln(1-4\epsilon^2)
    +
    \frac\alpha2\ln K
    -
    \frac12\ln 2A\right)+\ln (dC)\\
    &{\buildrel {(L_K\ge 1,  dC>1)}\over\le}&
     L_K\left(\frac{K}{4K_*}\ln(1-4\epsilon^2)
    +
   \frac\alpha2\ln K
    -
    \frac12\ln 2A+\ln (dC)\right)\\
    &\le&
    L_K\left(\frac{K}{4K_*}\ln(1-4\epsilon^2)
    +
   \frac\alpha2\left(\ln\frac{-4\alpha K_*/e}{\ln(1-4\epsilon^2)}-
   \frac{\ln(1-4\epsilon^2)}{4\alpha K_*}K\right)
    -
    \frac12\ln 2A+\ln \frac{8}{3\epsilon}+\ln d\right)\\
    &=&
\frac{L_K   \ln(1-4\epsilon^2)}{8} \left(\frac{K}{K_*}-c
+\frac{8\ln d}{\ln(1-4\epsilon^2)}\right)\le0
%\\    &=& 0 \hspace{50mm} \text{(Follows after the expression for \math{c}.)}
    }
%  Since
%  \math{c\ge-\frac{4}{\ln(1-4\epsilon^2)}\left(
%      \ln d+ \frac\alpha2\ln n
%    +\ln \frac{4\sqrt{2}}{3\sqrt A}\right)}, 
%  the latter expression in square parentheses is non-negative. And
%  \math{\ln(1-4\epsilon^2)<0}, which means that the entire expression is
%  non-positive, concluding the proof.
The last step follows from the choice of $c$.
\end{proof}

\subsection{Realizing a Consistent Representative Democracy}
In Theorem~\ref{thm:multiissue}, the group competency function $\rho$ is given as part of the input. In this section we focus on representative selection processes that would lead to a consistent representative democracy. We first start with an example on $F$ where no representative democracy is consistent. %Necessarily, such inconsistent
%representative democracies will have to violate monotonicity or
%consistency for at least one marginal single-issue induced representative
%democracy.

\begin{ex}\label{ex:multiinconsistent}\rm
  Let $F$ be the uniform distribution over two competence vectors: $\vec c_1$ from
  Example~\ref{ex:multi} and $\vec c_2=(\frac29,\frac19,\frac49,\frac29)$ for $(0,0), (0,1), (1,0), (1,1)$, respectively. The two issues are independent in both $\vec c_1$ ($\frac 13$ for $I_1$ and $\frac 23$ for $I_2$) and $\vec c_2$ ($\frac 23$ for $I_1$ and $\frac 13$ for $I_2$). We note that for each type of
  voter, the marginal competences of two issues always sum up to $1$.

Let $\rho$ denote the group competence function of any representative selection process that chooses a group member  as the representative. Then, the marginal competences of the chosen representative must sum up to $1$. Therefore, for any $K$, we must have $\mu_{1}(K)+\mu_{2}(K)=1$. Therefore, for any partition function $\vec K$ and any $n$, either the average marginal competence for $I_1$ is no more than $0.5$, or the average marginal competence for $I_2$ is no more than $0.5$, which means that the majority vote of $I_1$ or $I_2$ will not be correct
with probability $1$ as $n\rightarrow \infty$ (Lemma~\ref{lem:poisson-binomial}).
\hfill$\blacksquare$
\end{ex}
%\subsubsection{Consistent Representative Democracy For Independent Issues}

We see in Example~\ref{ex:multiinconsistent} that when the issues are correlated, sometimes nothing can
be done to get multi-issue consistency. A natural selection process is for  a group to choose the voter
having the maximum ``total competence'' over all the issues. This is the process used in Example~\ref{ex:multi}.
	
\begin{dfn}[Max-Sum Process]
Given a group of $K$ voters whose votes on $d$ issues are represented by random variables $v_1,\ldots,v_K$, the {\sc max-sum} representative selection process chooses a voter $i$ with maximum $\Exp[v_i]\cdot \vec 1$. Let $\rho_\text{ms}$ denote its group competence function.
\end{dfn}
$\rho_\text{ms}$ naturally extends the {\sc max} group competence function $\mu_{\max}$ defined in Example~\ref{example:uniform-max}. We will prove that $\rho_\text{ms}$ works well when the $d$ issues are independent, formally defined below. 
\begin{dfn}[Independent Issues] In this setting, each voter's competence vector is uniquely characterized by a vector of $d$ numbers $(p_1,\ldots,p_d)$, where $p_i$ is the probability that the voter's vote for issue $i$ is correct,
  independent of the other issues. We further assume that $F$ is a product distribution
\math{f_1(p_1)f_2(p_2)\cdots f_d(p_d)}, and each \math{f_j(\cdot)} has support on \math{(\frac12,1]}.\footnote{\math{M_i=\sup_p\{p:f_i(p)>0\}\ge \frac12+\epsilon} and \math{f_i} has nonzero support at
    \math{M_i}, that is 
\math{\Prob[p_j\ge M_i-\delta]>0}
for all \math{\delta>0}.}
\end{dfn}

\begin{thm}\label{thm:consistent-multiple}
 The max sum representative democracy is consistent for independent issues.
% Suppose voters are independent and further that the voter probabilities
%  \math{p_i} for getting issue \math{\nu_i} correct are independent. Suppose
%  probability \math{p_i} has distribution \math{f_i(\cdot)} which has
%  support on \math{(\frac12,1]}. Then, for the maximum sum representative
%    election process, there is a \math{K_*} for which the
%    multiple-issue democracy is consistent.
\end{thm}
\begin{proof}
  Based on the distributions \math{f_i(\cdot)}, let \math{M_i} be the
  maximum possible value attainable by \math{p_i}
  (\math{M_i\ge\frac12+\epsilon})
  and let \math{M=\sum_iM_i} be the maximum possible value for
  a voter's sum of probabilities \math{\sum_{i=1}^dp_i}. The basic idea in the
  proof is to show that for sufficiently large \math{K}, the representative's
  \math{\sum_i{p_i}} approaches \math{M}, which means that each probability
  must approach \math{M_i} (all above \math{\frac12+\epsilon}), and once
  that happens, the majority vote among many representatives will get
  all issues correct.

  For a voter, let \math{s=\sum_ip_i}. First, let us prove that
  the probability for \math{s} to be close to its maximum possible value
  \math{M} is
  large. Recall that \math{M_i\ge\frac 12+\epsilon}. Let
  \math{P_i=\Prob[p_i\ge M_i-\frac1{2d}\epsilon]>0} (because \math{f_i}
  has
  support on \math{(\frac12,1]}).
    Let \math{P=P_1\times\cdots\times P_d}. Then,
    \math{
      \{
      p_1\ge M_1-\frac1{2d}\epsilon\text{ and }
      p_2\ge M_2-\frac1{2d}\epsilon\text{ and }
      \ldots\text{ and }
      p_d\ge M_d-\frac1{2d}\epsilon
      \}
    }
    implies \math{\{s\ge M-\frac12\epsilon\}} and we have 
  $
      \Prob[s\ge M-{\textstyle\frac12}\epsilon]
      \ge
      P_1\times\cdots\times P_d=P>0$.
 
    Let us now consider \math{K} independent voters and the sums
    of probabilities \math{s_1,\ldots,s_K}. The representative
    \math{r} is picked as the voter with maximum sum, and we have
    \eqan{
      \Prob[\max_\ell s_\ell\ge M-{\textstyle\frac12}\epsilon]
      &=&
      1-\prod_\ell\Prob[s_\ell< M-{\textstyle\frac12}\epsilon]=
      1-\Prob[s< M-{\textstyle\frac12}\epsilon]^K
      \ge
      1-(1-P)^K.
    }
    Observe that \math{s\ge M-{\textstyle\frac12}\epsilon} implies
    \math{p_i\ge M_i-{\textstyle\frac12}\epsilon} for every issue
    \math{I_i}. Therefore, for the representative \math{r},

$\hfill
      \Prob[p_i\ge M_i-{\textstyle\frac12}\epsilon]
    \ge
    \Prob[s(r)\ge M-{\textstyle\frac12}\epsilon]\ge
    1-(1-P)^K.
\hfill$

    Since \math{M_i\ge\frac12+\epsilon},
    \math{p_i\ge M_i-{\textstyle\frac12}\epsilon}
    implies \math{p_i\ge\frac12+\frac12\epsilon}, and so $
      \Prob[p_i\ge{\textstyle\frac12+\frac12\epsilon}]\ge
      1-(1-P)^K$.
      
    The marginal single-issue competence function
    for issue \math{I_i} is just \math{\Exp[p_i]},

$\hfill
      \mu_i(K)
      =\Exp[p_i]\ge\Prob[p_i\ge{\textstyle\frac12+\frac12\epsilon}]\times
         {\textstyle(\frac12+\frac12\epsilon)}
         \ge  (1-(1-P)^K) \times
         {\textstyle(\frac12+\frac12\epsilon)}
\hfill$

    We thus have a lower bound for each marginal density which is
    monotonically increasing in \math{K}.
    Further, by setting \math{K> \log(\frac{\epsilon}{1+\epsilon})/\log(1-P)},
    we find that
    \math{\mu_i(K)>\frac12}. Therefore
    the marginal single-issue group competence functions are
    monotonic and consistent. By Theorem~\ref{thm:consistent-multiple},
    the multiple-issue
    group competence function is consistent. 
\end{proof}

\section{Summary and Future Work}
We set the mathematical foundation for studying the quality-quantity tradeoff
  in a representative democracy by introducing a mathematical framework for
studying representative democracy, and show that under general and natural
conditions, the optimal group size is constant when the cost of voting is a constant, and is $\Omega(n/\log n)$ when the cost and benefit are both polynomial.

There are many open questions and future directions under our framework. Can we extend our results to inhomogeneous representative selection processes, e.g.~different states use different processes to choose representatives? Does diversity in population help make better decisions? More generally, it would be interesting and important to consider similar extensions as done for the Condorcet Jury Theorem, for example to  inhomogeneous agents and strategic agents.

%\input{interesting.tex}

%\input{inference.tex}

%\bibliographystyle{plainnat}
%\bibliographystyle{alpha}
%\bibliographystyle{ACM-Reference-Format}
%\bibliography{../references}

\begin{thebibliography}{00}

%%% ====================================================================
%%% NOTE TO THE USER: you can override these defaults by providing
%%% customized versions of any of these macros before the \bibliography
%%% command.  Each of them MUST provide its own final punctuation,
%%% except for \shownote{}, \showDOI{}, and \showURL{}.  The latter two
%%% do not use final punctuation, in order to avoid confusing it with
%%% the Web address.
%%%
%%% To suppress output of a particular field, define its macro to expand
%%% to an empty string, or better, \unskip, like this:
%%%
%%% \newcommand{\showDOI}[1]{\unskip}   % LaTeX syntax
%%%
%%% \def \showDOI #1{\unskip}           % plain TeX syntax
%%%
%%% ====================================================================

\ifx \showCODEN    \undefined \def \showCODEN     #1{\unskip}     \fi
\ifx \showDOI      \undefined \def \showDOI       #1{{\tt DOI:}\penalty0{#1}\ }
  \fi
\ifx \showISBNx    \undefined \def \showISBNx     #1{\unskip}     \fi
\ifx \showISBNxiii \undefined \def \showISBNxiii  #1{\unskip}     \fi
\ifx \showISSN     \undefined \def \showISSN      #1{\unskip}     \fi
\ifx \showLCCN     \undefined \def \showLCCN      #1{\unskip}     \fi
\ifx \shownote     \undefined \def \shownote      #1{#1}          \fi
\ifx \showarticletitle \undefined \def \showarticletitle #1{#1}   \fi
\ifx \showURL      \undefined \def \showURL       #1{#1}          \fi
% The following commands are used for tagged output and should be
% invisible to TeX
\providecommand\bibfield[2]{#2}
\providecommand\bibinfo[2]{#2}
\providecommand\natexlab[1]{#1}
\providecommand\showeprint[2][]{arXiv:#2}

\bibitem[\protect\citeauthoryear{Auriol and Gary-Bobo}{Auriol and
  Gary-Bobo}{2012}]%
        {Auriol2012:On-the-optimal}
\bibfield{author}{\bibinfo{person}{Emmanuelle Auriol} {and}
  \bibinfo{person}{Robert~J. Gary-Bobo}.} \bibinfo{year}{2012}\natexlab{}.
\newblock \showarticletitle{On the optimal number of representatives}.
\newblock \bibinfo{journal}{{\em Public Choice\/}} \bibinfo{volume}{153},
  \bibinfo{number}{3--4} (\bibinfo{year}{2012}), \bibinfo{pages}{419--445}.
\newblock


\bibitem[\protect\citeauthoryear{Azari~Soufiani, Parkes, and
  Xia}{Azari~Soufiani et~al\mbox{.}}{2014}]%
        {Azari14:Statistical}
\bibfield{author}{\bibinfo{person}{Hossein Azari~Soufiani},
  \bibinfo{person}{David~C. Parkes}, {and} \bibinfo{person}{Lirong Xia}.}
  \bibinfo{year}{2014}\natexlab{}.
\newblock \showarticletitle{Statistical Decision Theory Approaches to Social
  Choice}. In \bibinfo{booktitle}{{\em Proceedings of Advances in Neural
  Information Processing Systems (NIPS)}}. \bibinfo{address}{Montreal, Quebec,
  Canada}.
\newblock


\bibitem[\protect\citeauthoryear{Bartlett}{Bartlett}{2014}]%
        {NYTenlarge}
\bibfield{author}{\bibinfo{person}{Bruce Bartlett}.}
  \bibinfo{year}{2014}\natexlab{}.
\newblock \bibinfo{title}{{Enlarging the House of Representatives}}.
\newblock
  \bibinfo{howpublished}{\url{https://economix.blogs.nytimes.com/2014/01/07/enlarging-the-house-of-representatives/}}.
    (\bibinfo{year}{2014}).
\newblock


\bibitem[\protect\citeauthoryear{Ben-Yashar and Paroush}{Ben-Yashar and
  Paroush}{2000}]%
        {Ben00:Nonasymptotic}
\bibfield{author}{\bibinfo{person}{Ruth Ben-Yashar} {and}
  \bibinfo{person}{Jacob Paroush}.} \bibinfo{year}{2000}\natexlab{}.
\newblock \showarticletitle{{A nonasymptotic Condorcet jury theorem}}.
\newblock \bibinfo{journal}{{\em Social Choice and Welfare\/}}
  \bibinfo{volume}{17}, \bibinfo{number}{2} (\bibinfo{year}{2000}),
  \bibinfo{pages}{189--199}.
\newblock


\bibitem[\protect\citeauthoryear{Ben-Yashar and Paroush}{Ben-Yashar and
  Paroush}{2003}]%
        {Ben03:Investment}
\bibfield{author}{\bibinfo{person}{Ruth Ben-Yashar} {and}
  \bibinfo{person}{Jacob Paroush}.} \bibinfo{year}{2003}\natexlab{}.
\newblock \showarticletitle{{Investment in Human Capital in Team Members Who
  Are Involved in Collective Decision Making}}.
\newblock \bibinfo{journal}{{\em Journal of Public Economic Theory\/}}
  \bibinfo{volume}{5}, \bibinfo{number}{3} (\bibinfo{year}{2003}),
  \bibinfo{pages}{527---539}.
\newblock


\bibitem[\protect\citeauthoryear{Ben-Yashar and Zahavi}{Ben-Yashar and
  Zahavi}{2011}]%
        {Ben11:Condorcet}
\bibfield{author}{\bibinfo{person}{Ruth Ben-Yashar} {and} \bibinfo{person}{Mor
  Zahavi}.} \bibinfo{year}{2011}\natexlab{}.
\newblock \showarticletitle{{The Condorcet jury theorem and extension of the
  franchise with rationally ignorant voters}}.
\newblock \bibinfo{journal}{{\em Public Choice\/}} \bibinfo{volume}{148},
  \bibinfo{number}{3} (\bibinfo{year}{2011}), \bibinfo{pages}{435--443}.
\newblock


\bibitem[\protect\citeauthoryear{Berend and Paroush}{Berend and
  Paroush}{1998}]%
        {Berend98:When}
\bibfield{author}{\bibinfo{person}{Daniel Berend} {and} \bibinfo{person}{Jacob
  Paroush}.} \bibinfo{year}{1998}\natexlab{}.
\newblock \showarticletitle{{When is Condorcet's Jury Theorem valid?}}
\newblock \bibinfo{journal}{{\em Social Choice and Welfare\/}}
  \bibinfo{volume}{15}, \bibinfo{number}{4} (\bibinfo{year}{1998}),
  \bibinfo{pages}{481--488}.
\newblock


\bibitem[\protect\citeauthoryear{Berend and Sapir}{Berend and Sapir}{2005}]%
        {Berend05:Monotonicity}
\bibfield{author}{\bibinfo{person}{Daniel Berend} {and} \bibinfo{person}{Luba
  Sapir}.} \bibinfo{year}{2005}\natexlab{}.
\newblock \showarticletitle{{Monotonicity in Condorcet Jury Theorem}}.
\newblock \bibinfo{journal}{{\em Social Choice and Welfare\/}}
  \bibinfo{volume}{24} (\bibinfo{year}{2005}), \bibinfo{pages}{83--92}.
\newblock


\bibitem[\protect\citeauthoryear{Berend and Sapir}{Berend and Sapir}{2007}]%
        {Berend07:Monotonicity}
\bibfield{author}{\bibinfo{person}{Daniel Berend} {and} \bibinfo{person}{Luba
  Sapir}.} \bibinfo{year}{2007}\natexlab{}.
\newblock \showarticletitle{{Monotonicity in Condorcet's Jury Theorem with
  dependent voters}}.
\newblock \bibinfo{journal}{{\em Social Choice and Welfare\/}}
  \bibinfo{volume}{28}, \bibinfo{number}{3} (\bibinfo{year}{2007}),
  \bibinfo{pages}{507--528}.
\newblock


\bibitem[\protect\citeauthoryear{Besley and Coate}{Besley and Coate}{1997}]%
        {Besley1997:An-Economic}
\bibfield{author}{\bibinfo{person}{Timothy Besley} {and}
  \bibinfo{person}{Stephen Coate}.} \bibinfo{year}{1997}\natexlab{}.
\newblock \showarticletitle{{An Economic Model of Representative Democracy}}.
\newblock \bibinfo{journal}{{\em The Quarterly Journal of Economics\/}}
  \bibinfo{volume}{122}, \bibinfo{number}{1} (\bibinfo{year}{1997}),
  \bibinfo{pages}{85--114}.
\newblock


\bibitem[\protect\citeauthoryear{Caragiannis, Procaccia, and Shah}{Caragiannis
  et~al\mbox{.}}{2016}]%
        {Caragiannis2016:When}
\bibfield{author}{\bibinfo{person}{Ioannis Caragiannis},
  \bibinfo{person}{Ariel~D. Procaccia}, {and} \bibinfo{person}{Nisarg Shah}.}
  \bibinfo{year}{2016}\natexlab{}.
\newblock \showarticletitle{{When Do Noisy Votes Reveal the Truth?}}
\newblock \bibinfo{journal}{{\em ACM Transactions on Economics and
  Computation\/}} \bibinfo{volume}{4}, \bibinfo{number}{3}
  (\bibinfo{year}{2016}), \bibinfo{pages}{Article No.~15}.
\newblock


\bibitem[\protect\citeauthoryear{Cohensius, Mannor, Meir, Meirom, and
  Orda}{Cohensius et~al\mbox{.}}{2017}]%
        {Cohensius2017:Proxy}
\bibfield{author}{\bibinfo{person}{Gal Cohensius}, \bibinfo{person}{Shie
  Mannor}, \bibinfo{person}{Reshef Meir}, \bibinfo{person}{Eli Meirom}, {and}
  \bibinfo{person}{Ariel Orda}.} \bibinfo{year}{2017}\natexlab{}.
\newblock \showarticletitle{{Proxy Voting for Better Outcomes}}. In
  \bibinfo{booktitle}{{\em Proceedings of the 16th Conference on Autonomous
  Agents and MultiAgent Systems}}. \bibinfo{pages}{858--866}.
\newblock


\bibitem[\protect\citeauthoryear{Condorcet}{Condorcet}{1785}]%
        {Condorcet1785:Essai}
\bibfield{author}{\bibinfo{person}{Marquis~de Condorcet}.}
  \bibinfo{year}{1785}\natexlab{}.
\newblock \bibinfo{booktitle}{{\em Essai sur l'application de l'analyse \`a la
  probabilit\'e des d\'ecisions rendues \`a la pluralit\'e des voix}}.
\newblock Paris: L'Imprimerie Royale.
\newblock


\bibitem[\protect\citeauthoryear{Conitzer, Freeman, and Shah}{Conitzer
  et~al\mbox{.}}{2017}]%
        {Conitzer2017:Fair}
\bibfield{author}{\bibinfo{person}{Vincent Conitzer}, \bibinfo{person}{Rupert
  Freeman}, {and} \bibinfo{person}{Nisarg Shah}.}
  \bibinfo{year}{2017}\natexlab{}.
\newblock \showarticletitle{{Fair Public Decision Making}}. In
  \bibinfo{booktitle}{{\em Proceedings of the 18th ACM Conference on Economics
  and Computation}}.
\newblock


\bibitem[\protect\citeauthoryear{Conitzer and Sandholm}{Conitzer and
  Sandholm}{2005}]%
        {Conitzer05:Common}
\bibfield{author}{\bibinfo{person}{Vincent Conitzer} {and}
  \bibinfo{person}{Tuomas Sandholm}.} \bibinfo{year}{2005}\natexlab{}.
\newblock \showarticletitle{Common Voting Rules as Maximum Likelihood
  Estimators}. In \bibinfo{booktitle}{{\em Proceedings of the 21st Annual
  Conference on Uncertainty in Artificial Intelligence (UAI)}}.
  \bibinfo{address}{Edinburgh, UK}, \bibinfo{pages}{145--152}.
\newblock


\bibitem[\protect\citeauthoryear{Elkind and Shah}{Elkind and Shah}{2014}]%
        {Elkind14:Electing}
\bibfield{author}{\bibinfo{person}{Edith Elkind} {and} \bibinfo{person}{Nisarg
  Shah}.} \bibinfo{year}{2014}\natexlab{}.
\newblock \showarticletitle{{Electing the Most Probable Without Eliminating the
  Irrational: Voting Over Intransitive Domains}}. In \bibinfo{booktitle}{{\em
  Proceedings of the 30th Conference on Uncertainty in Artificial
  Intelligence}}. \bibinfo{pages}{182--191}.
\newblock


\bibitem[\protect\citeauthoryear{Feld and Grofman}{Feld and Grofman}{1984}]%
        {Feld84:Accuracy}
\bibfield{author}{\bibinfo{person}{Scott~L. Feld} {and}
  \bibinfo{person}{Bernard Grofman}.} \bibinfo{year}{1984}\natexlab{}.
\newblock \showarticletitle{The accuracy of group majority decisions in groups
  with added members}.
\newblock \bibinfo{journal}{{\em Public Choice\/}} \bibinfo{volume}{42},
  \bibinfo{number}{3} (\bibinfo{year}{1984}), \bibinfo{pages}{273--285}.
\newblock


\bibitem[\protect\citeauthoryear{Fey}{Fey}{2003}]%
        {Fey03:Note}
\bibfield{author}{\bibinfo{person}{Mark Fey}.} \bibinfo{year}{2003}\natexlab{}.
\newblock \showarticletitle{{A note on the Condorcet Jury Theorem with
  supermajority voting rules}}.
\newblock \bibinfo{journal}{{\em Social Choice and Welfare\/}}
  \bibinfo{volume}{20}, \bibinfo{number}{1} (\bibinfo{year}{2003}),
  \bibinfo{pages}{27--32}.
\newblock


\bibitem[\protect\citeauthoryear{Flynn}{Flynn}{2012}]%
        {CNNexpandcongress}
\bibfield{author}{\bibinfo{person}{Brian Flynn}.}
  \bibinfo{year}{2012}\natexlab{}.
\newblock \bibinfo{title}{{What's wrong with Congress? It's not big enough}}.
\newblock
  \bibinfo{howpublished}{\url{https://www.cnn.com/2012/03/09/opinion/flynn-expand-congress/index.html}}.
    (\bibinfo{year}{2012}).
\newblock


\bibitem[\protect\citeauthoryear{Gentle}{Gentle}{2009}]%
        {Gentle2009:Computational}
\bibfield{author}{\bibinfo{person}{James~E. Gentle}.}
  \bibinfo{year}{2009}\natexlab{}.
\newblock \bibinfo{booktitle}{{\em {Computational Statistics}}}.
\newblock \bibinfo{publisher}{Springer}.
\newblock


\bibitem[\protect\citeauthoryear{Gradstein and Nitzan}{Gradstein and
  Nitzan}{1987}]%
        {Gradstein87:Organizational}
\bibfield{author}{\bibinfo{person}{Mark Gradstein} {and}
  \bibinfo{person}{Shmuel Nitzan}.} \bibinfo{year}{1987}\natexlab{}.
\newblock \showarticletitle{Organizational decision-making quality and the
  severity of the free-riding problem}.
\newblock \bibinfo{journal}{{\em Economics Letters\/}} \bibinfo{volume}{23},
  \bibinfo{number}{4} (\bibinfo{year}{1987}), \bibinfo{pages}{335--339}.
\newblock


\bibitem[\protect\citeauthoryear{Grofman, Owen, and Feld}{Grofman
  et~al\mbox{.}}{1983}]%
        {Grofman83:Thirteen}
\bibfield{author}{\bibinfo{person}{Bernard Grofman}, \bibinfo{person}{Guillermo
  Owen}, {and} \bibinfo{person}{Scott~L. Feld}.}
  \bibinfo{year}{1983}\natexlab{}.
\newblock \showarticletitle{Thirteen theorems in search of the truth}.
\newblock \bibinfo{journal}{{\em Theory and Decision\/}} \bibinfo{volume}{15},
  \bibinfo{number}{3} (\bibinfo{year}{1983}), \bibinfo{pages}{261--278}.
\newblock


\bibitem[\protect\citeauthoryear{Hoeffding}{Hoeffding}{1956}]%
        {Hoeffding1956:On-the-Distribution}
\bibfield{author}{\bibinfo{person}{Wassily Hoeffding}.}
  \bibinfo{year}{1956}\natexlab{}.
\newblock \showarticletitle{{On the Distribution of the Number of Successes in
  Independent Trials}}.
\newblock \bibinfo{journal}{{\em The Annals of Mathematical Statistics\/}}
  \bibinfo{volume}{27}, \bibinfo{number}{3} (\bibinfo{year}{1956}),
  \bibinfo{pages}{713--721}.
\newblock


\bibitem[\protect\citeauthoryear{Humphreys}{Humphreys}{2016}]%
        {WashingtonPexpand}
\bibfield{author}{\bibinfo{person}{Keith Humphreys}.}
  \bibinfo{year}{2016}\natexlab{}.
\newblock \bibinfo{title}{{Why we might want to grow the House of
  Representatives by 250 more seats}}.
\newblock
  \bibinfo{howpublished}{\url{https://www.washingtonpost.com/news/wonk/wp/2016/10/03/why-we-might-want-to-grow-the-house-of-representatives-by-250-more-seats/?utm_term=.2948bef5e749}}.
    (\bibinfo{year}{2016}).
\newblock


\bibitem[\protect\citeauthoryear{Kahng, Mackenzie, and Procaccia}{Kahng
  et~al\mbox{.}}{2018}]%
        {Kahng2018:Liquid}
\bibfield{author}{\bibinfo{person}{Anson Kahng}, \bibinfo{person}{Simon
  Mackenzie}, {and} \bibinfo{person}{Ariel~D. Procaccia}.}
  \bibinfo{year}{2018}\natexlab{}.
\newblock \showarticletitle{{Liquid Democracy: An Algorithmic Perspective.}}.
  In \bibinfo{booktitle}{{\em Proc. 32nd AAAI Conference on Artificial
  Intelligence}}.
\newblock


\bibitem[\protect\citeauthoryear{Kanazawa}{Kanazawa}{1998}]%
        {Kanazawa98:Brief}
\bibfield{author}{\bibinfo{person}{Satoshi Kanazawa}.}
  \bibinfo{year}{1998}\natexlab{}.
\newblock \showarticletitle{{A brief note on a further refinement of the
  Condorcet Jury Theorem for heterogeneous groups}}.
\newblock \bibinfo{journal}{{\em Mathematical Social Sciences\/}}
  \bibinfo{volume}{35}, \bibinfo{number}{1} (\bibinfo{year}{1998}),
  \bibinfo{pages}{69--73}.
\newblock


\bibitem[\protect\citeauthoryear{Karotkin and Paroush}{Karotkin and
  Paroush}{1995}]%
        {Karotkin95:Incentive}
\bibfield{author}{\bibinfo{person}{Drora Karotkin} {and} \bibinfo{person}{Jacob
  Paroush}.} \bibinfo{year}{1995}\natexlab{}.
\newblock \showarticletitle{Incentive schemes for investment in human capital
  by members of a team of decision makers}.
\newblock \bibinfo{journal}{{\em Labour Economics\/}} \bibinfo{volume}{2},
  \bibinfo{number}{1} (\bibinfo{year}{1995}), \bibinfo{pages}{41---51}.
\newblock


\bibitem[\protect\citeauthoryear{Karotkin and Paroush}{Karotkin and
  Paroush}{2003}]%
        {Karotkin03:Optimum}
\bibfield{author}{\bibinfo{person}{Drora Karotkin} {and} \bibinfo{person}{Jacob
  Paroush}.} \bibinfo{year}{2003}\natexlab{}.
\newblock \showarticletitle{{Optimum committee size: Quality-versus-quantity
  dilemma}}.
\newblock \bibinfo{journal}{{\em Social Choice and Welfare\/}}
  \bibinfo{volume}{20} (\bibinfo{year}{2003}), \bibinfo{pages}{429--441}.
\newblock


\bibitem[\protect\citeauthoryear{Koriyama, Mac{\'e}, Treibich, and
  Laslier}{Koriyama et~al\mbox{.}}{2013}]%
        {Koriyama2013:Optimal}
\bibfield{author}{\bibinfo{person}{Yukio Koriyama}, \bibinfo{person}{Antonin
  Mac{\'e}}, \bibinfo{person}{Rafael Treibich}, {and}
  \bibinfo{person}{Jean-Fran{\c c}ois Laslier}.}
  \bibinfo{year}{2013}\natexlab{}.
\newblock \showarticletitle{{Optimal Apportionment}}.
\newblock \bibinfo{journal}{{\em Journal of Political Economy\/}}
  \bibinfo{volume}{121}, \bibinfo{number}{3} (\bibinfo{year}{2013}),
  \bibinfo{pages}{584--608}.
\newblock


\bibitem[\protect\citeauthoryear{Lang and Xia}{Lang and Xia}{2016}]%
        {Lang16:Voting}
\bibfield{author}{\bibinfo{person}{J\'{e}r\^{o}me Lang} {and}
  \bibinfo{person}{Lirong Xia}.} \bibinfo{year}{2016}\natexlab{}.
\newblock \showarticletitle{{Voting in Combinatorial Domains}}.
\newblock In \bibinfo{booktitle}{{\em {Handbook of Computational Social
  Choice}}}, \bibfield{editor}{\bibinfo{person}{Felix Brandt},
  \bibinfo{person}{Vincent Conitzer}, \bibinfo{person}{Ulle Endriss},
  \bibinfo{person}{J\'{e}r\^{o}me Lang}, {and} \bibinfo{person}{Ariel
  Procaccia}} (Eds.). \bibinfo{publisher}{Cambridge University Press},
  Chapter~9.
\newblock


\bibitem[\protect\citeauthoryear{Miller}{Miller}{1986}]%
        {Miller86:Information}
\bibfield{author}{\bibinfo{person}{Nicholas~R. Miller}.}
  \bibinfo{year}{1986}\natexlab{}.
\newblock \showarticletitle{{Information, Electorates, and Democracy: Some
  Extensions and Interpretations of the Condorcet Jury Theorem}}.
\newblock In \bibinfo{booktitle}{{\em Information Pooling and Group Decision
  Making}}, \bibfield{editor}{\bibinfo{person}{Grofman B.} {and}
  \bibinfo{person}{Owen G.}} (Eds.). \bibinfo{publisher}{JAI Press},
  \bibinfo{pages}{173---192}.
\newblock


\bibitem[\protect\citeauthoryear{Mukhopadhaya}{Mukhopadhaya}{2003}]%
        {Mukhopadhaya03:Jury}
\bibfield{author}{\bibinfo{person}{Kaushik Mukhopadhaya}.}
  \bibinfo{year}{2003}\natexlab{}.
\newblock \showarticletitle{{Jury Size and the Free Rider Problem}}.
\newblock \bibinfo{journal}{{\em Journal of Law, Economics, and
  Organization\/}} \bibinfo{volume}{19}, \bibinfo{number}{1}
  (\bibinfo{year}{2003}), \bibinfo{pages}{24--44}.
\newblock


\bibitem[\protect\citeauthoryear{Nitzan and Paroush}{Nitzan and
  Paroush}{1980}]%
        {Nitzan80:Investment}
\bibfield{author}{\bibinfo{person}{Shmuel Nitzan} {and} \bibinfo{person}{Jacob
  Paroush}.} \bibinfo{year}{1980}\natexlab{}.
\newblock \showarticletitle{{Investment in Human Capital and Social Self
  Protection under Uncertainty}}.
\newblock \bibinfo{journal}{{\em International Economic Review\/}}
  \bibinfo{volume}{21}, \bibinfo{number}{3} (\bibinfo{year}{1980}),
  \bibinfo{pages}{547--557}.
\newblock


\bibitem[\protect\citeauthoryear{Nitzan and Paroush}{Nitzan and
  Paroush}{1984}]%
        {Nitzan84:Significance}
\bibfield{author}{\bibinfo{person}{Shmuel Nitzan} {and} \bibinfo{person}{Jacob
  Paroush}.} \bibinfo{year}{1984}\natexlab{}.
\newblock \showarticletitle{{The significance of independent decisions in
  uncertain dichotomous choice situations}}.
\newblock \bibinfo{journal}{{\em Theory and Decision\/}} \bibinfo{volume}{17},
  \bibinfo{number}{1} (\bibinfo{year}{1984}), \bibinfo{pages}{47--60}.
\newblock


\bibitem[\protect\citeauthoryear{Nitzan and Paroush}{Nitzan and
  Paroush}{2017}]%
        {Nitzan17:Collective}
\bibfield{author}{\bibinfo{person}{Shmuel Nitzan} {and} \bibinfo{person}{Jacob
  Paroush}.} \bibinfo{year}{2017}\natexlab{}.
\newblock \showarticletitle{Collective Decision Making and Jury Theorems}.
\newblock In \bibinfo{booktitle}{{\em {The Oxford Handbook of Law and
  Economics: Volume 1: Methodology and Concepts}}},
  \bibfield{editor}{\bibinfo{person}{Francesco Parisi}} (Ed.).
  \bibinfo{publisher}{Oxford University Press}.
\newblock


\bibitem[\protect\citeauthoryear{Owen, Grofman, and Feld}{Owen
  et~al\mbox{.}}{1989}]%
        {Owen89:Proving}
\bibfield{author}{\bibinfo{person}{Guillermo Owen}, \bibinfo{person}{Bernard
  Grofman}, {and} \bibinfo{person}{Scott~L. Feld}.}
  \bibinfo{year}{1989}\natexlab{}.
\newblock \showarticletitle{{Proving a distribution-free generalization of the
  Condorcet Jury Theorem}}.
\newblock \bibinfo{journal}{{\em Mathematical Social Sciences\/}}
  \bibinfo{volume}{17}, \bibinfo{number}{1} (\bibinfo{year}{1989}),
  \bibinfo{pages}{1--16}.
\newblock


\bibitem[\protect\citeauthoryear{Paroush}{Paroush}{1998}]%
        {Paroush98:Stay}
\bibfield{author}{\bibinfo{person}{Jacob Paroush}.}
  \bibinfo{year}{1998}\natexlab{}.
\newblock \showarticletitle{{Stay away from fair coins: A Condorcet jury
  theorem}}.
\newblock \bibinfo{journal}{{\em Social Choice and Welfare\/}}
  \bibinfo{volume}{15}, \bibinfo{number}{1} (\bibinfo{year}{1998}),
  \bibinfo{pages}{15--20}.
\newblock


\bibitem[\protect\citeauthoryear{Paroush and Karotkin}{Paroush and
  Karotkin}{1989}]%
        {Paroush89:Robustness}
\bibfield{author}{\bibinfo{person}{Jacob Paroush} {and} \bibinfo{person}{D.
  Karotkin}.} \bibinfo{year}{1989}\natexlab{}.
\newblock \showarticletitle{{Robustness of Optimal Majority Rules Over Teams
  with Changing Size}}.
\newblock \bibinfo{journal}{{\em Social Choice and Welfare\/}}
  \bibinfo{volume}{6}, \bibinfo{number}{2} (\bibinfo{year}{1989}),
  \bibinfo{pages}{127--138}.
\newblock


\bibitem[\protect\citeauthoryear{Pivato}{Pivato}{2013}]%
        {Pivato13:Voting}
\bibfield{author}{\bibinfo{person}{Marcus Pivato}.}
  \bibinfo{year}{2013}\natexlab{}.
\newblock \showarticletitle{Voting rules as statistical estimators}.
\newblock \bibinfo{journal}{{\em Social Choice and Welfare\/}}
  \bibinfo{volume}{40}, \bibinfo{number}{2} (\bibinfo{year}{2013}),
  \bibinfo{pages}{581--630}.
\newblock


\bibitem[\protect\citeauthoryear{Procaccia, Reddi, and Shah}{Procaccia
  et~al\mbox{.}}{2012}]%
        {Procaccia12:Maximum}
\bibfield{author}{\bibinfo{person}{Ariel~D. Procaccia},
  \bibinfo{person}{Sashank~J. Reddi}, {and} \bibinfo{person}{Nisarg Shah}.}
  \bibinfo{year}{2012}\natexlab{}.
\newblock \showarticletitle{A Maximum Likelihood Approach For Selecting Sets of
  Alternatives.}. In \bibinfo{booktitle}{{\em Proceedings of the 28th
  Conference on Uncertainty in Artificial Intelligence}}.
\newblock


\bibitem[\protect\citeauthoryear{Sapir}{Sapir}{2005}]%
        {Sapir05:Generalized}
\bibfield{author}{\bibinfo{person}{Luba Sapir}.}
  \bibinfo{year}{2005}\natexlab{}.
\newblock \showarticletitle{Generalized means of jurors' competencies and
  marginal changes of jury's size}.
\newblock \bibinfo{journal}{{\em Mathematical Social Sciences\/}}
  \bibinfo{volume}{50}, \bibinfo{number}{1} (\bibinfo{year}{2005}),
  \bibinfo{pages}{83--101}.
\newblock


\bibitem[\protect\citeauthoryear{Skowron}{Skowron}{2015}]%
        {Skowron2015:What}
\bibfield{author}{\bibinfo{person}{Piotr Skowron}.}
  \bibinfo{year}{2015}\natexlab{}.
\newblock \showarticletitle{{What Do We Elect Committees For? A Voting
  Committee Model for Multi-Winner Rules}}. In \bibinfo{booktitle}{{\em
  Proceedings of the Twenty-Fourth International Joint Conference on Artificial
  Intelligence}}. \bibinfo{pages}{1141--1147}.
\newblock


\bibitem[\protect\citeauthoryear{Stone and Kagotani}{Stone and
  Kagotani}{2013}]%
        {Stone2013:Optimal}
\bibfield{author}{\bibinfo{person}{Peter Stone} {and} \bibinfo{person}{Koji
  Kagotani}.} \bibinfo{year}{2013}\natexlab{}.
\newblock \bibinfo{title}{{Optimal Committee Performance: Size versus
  Diversity}}.
\newblock \bibinfo{howpublished}{draft}.   (\bibinfo{year}{2013}).
\newblock


\bibitem[\protect\citeauthoryear{Xia}{Xia}{2016}]%
        {Xia2016:Bayesian}
\bibfield{author}{\bibinfo{person}{Lirong Xia}.}
  \bibinfo{year}{2016}\natexlab{}.
\newblock \showarticletitle{{Bayesian estimators as voting rules}}. In
  \bibinfo{booktitle}{{\em Proceedings of the Thirty-Second Conference on
  Uncertainty in Artificial Intelligence}}. \bibinfo{pages}{785--794}.
\newblock


\bibitem[\protect\citeauthoryear{Xia and Conitzer}{Xia and Conitzer}{2011}]%
        {Xia11:Maximum}
\bibfield{author}{\bibinfo{person}{Lirong Xia} {and} \bibinfo{person}{Vincent
  Conitzer}.} \bibinfo{year}{2011}\natexlab{}.
\newblock \showarticletitle{A Maximum Likelihood Approach towards Aggregating
  Partial Orders}. In \bibinfo{booktitle}{{\em Proceedings of the Twenty-Second
  International Joint Conference on Artificial Intelligence (IJCAI)}}.
  \bibinfo{address}{Barcelona, Catalonia, Spain}, \bibinfo{pages}{446--451}.
\newblock


\bibitem[\protect\citeauthoryear{Young}{Young}{1988}]%
        {Young88:Condorcet}
\bibfield{author}{\bibinfo{person}{H.~Peyton Young}.}
  \bibinfo{year}{1988}\natexlab{}.
\newblock \showarticletitle{Condorcet's Theory of Voting}.
\newblock \bibinfo{journal}{{\em American Political Science Review\/}}
  \bibinfo{volume}{82} (\bibinfo{year}{1988}), \bibinfo{pages}{1231--1244}.
\newblock


\end{thebibliography}
%\input{old.tex}
%%% -*-BibTeX-*-
%%% Do NOT edit. File created by BibTeX with style
%%% ACM-Reference-Format-Journals [18-Jan-2012].

\end{document}